\documentclass[11pt]{article}

\usepackage[utf8]{inputenc}
\pdfoutput=1
\usepackage[margin=1.1in,letterpaper]{geometry}

\usepackage{bm}

\makeatletter
\newlength\aftertitskip     \newlength\beforetitskip
\newlength\interauthorskip  \newlength\aftermaketitskip

\makeatother

\pdfoutput=1                    

\usepackage{color}
\usepackage{mathtools}
\usepackage{amsmath}
\usepackage{amsthm}
\usepackage{amssymb}
\usepackage{mathdots}
\usepackage{graphicx}
\usepackage[unicode=true,pdfusetitle,
bookmarks=true,bookmarksnumbered=false,bookmarksopen=false,
breaklinks=false,pdfborder={0 0 1},backref=false,colorlinks=true]
{hyperref}
\hypersetup{
	urlcolor=blue,citecolor=black}

\makeatletter

\newtheorem{theorem}{Theorem}[section]
\newtheorem{lemma}[theorem]{Lemma}
\newtheorem{corollary}[theorem]{Corollary}
\newtheorem{proposition}[theorem]{Proposition}

\theoremstyle{definition}
\newtheorem{definition}[theorem]{Definition}
\newtheorem{remark}[theorem]{Remark}

\numberwithin{equation}{section}

\title{Balance with Memory in Signed Networks via Mittag-Leffler Matrix Functions}
\author{Yu Tian
	\thanks{Nordita, Stockholm University and KTH Royal Institute of Technology, SE-106 91 Stockholm, Sweden (yu.tian@su.se).}
	\and Ernesto Estrada
	\thanks{Institute of Cross-Disciplinary Physics and Complex Systems, IFISC (UIB-CSIC), Palma de Mallorca, 07122, Spain (estrada@ifisc.uib-csic.es).}
}
\date{}

\begin{document}
	\maketitle

	\begin{abstract}
		Structural balance is an important characteristic of graphs/networks where edges can be positive or negative, with direct impact on the study of real-world complex systems. When a network is not structurally balanced, it is important to know how much balance still exists in it. Although several measures have been proposed to characterize the degree of balance, the use of matrix functions of the signed adjacency matrix emerges as a very promising area of research. 
		Here, we take a step forward to using Mittag-Leffler (ML) matrix functions to quantify the notion of balance of signed networks. We show that the ML balance index can be obtained from first principles on the basis of a nonconservative diffusion dynamic, and that it accounts for the memory of the system about the past, by diminishing the penalization that long cycles typically receive in other matrix functions. 
		Finally, we demonstrate the important information in the ML balance index with both artificial signed networks and real-world networks in various contexts, ranging from biological and ecological to social ones.
	\end{abstract}
	
	
	
	\section{Introduction}
	The use of matrix functions \cite{higham2008functions} has represented a significant advance in the development of mathematical models of networks $G=\left(V,E\right)$ in the last 20 years \cite{benzi2020matrix,estrada2010network}. In particular, the use of functions of the adjacency matrix $A$ of a network, $f\left(A\right)$, has impacted the areas of study of vertex centrality measures \cite{benzi2020matrix} as well as our understanding of the navigability of networks \cite{estrada2023network,seguin2018navigation}. 
	The mathematical roots of these developments come from the fact that $\left(A^{k}\right)_{uv}$ counts the number of walks of length $k$ connecting the vertices $u,v\in V$, where a walk is a sequence of (not necessarily different) consecutive vertices and edges in the network (see \cite{berge2001theory,festinger1949analysis,katz1953new} for original sources). Therefore, defining matrix functions of the type $f\left(A\right)=\sum_{k=0}^{\infty}c_{k}A^{k}$ allows to quantify the ``importance'', or centrality, of a vertex $u\in V$, by taking $\left(f\left(A\right)\right)_{uu}$ as a counting of all self-returning walks starting at vertex $u$, and giving more weight to the smaller than to the longer ones through the constants $\{c_k\}$ \cite{estrada2005subgraph}. Similarly, the term $\left(f\left(A\right)\right)_{uv}$ accounts for the ``communicability'' capacity between the vertices \cite{estrada2008communicability}. 
	Building on the field of Euclidean matrix theory \cite{balaji2007euclidean,gower1985properties,krislock2012euclidean}, circum-Euclidean distances \cite{alfakih2018euclidean,tarazaga1996circum}, a.k.a, spherical Euclidean distance, between pairs of vertices can also be obtained by defining $\left(f\left(A\right)\right)_{uu}+\left(f\left(A\right)\right)_{vv}-\left(f\left(A\right)\right)_{uv}$ for positive-definite matrix functions $f\left(A\right)$ \cite{estrada2012communicability} and angles $\left(f\left(A\right)\right)_{uv}/\sqrt{\left(f\left(A\right)\right)_{uu}\left(f\left(A\right)\right)_{vv}}$ \cite{estrada2016communicability} (see also \cite{estrada2014hyperspherical}).
	
	The historical background for the use of matrix functions to study networks can be traced back to the work of Katz \cite{katz1953new} who proposed $\left(I-\varepsilon A\right)^{-1}$, with $0<\varepsilon<\left(\lambda_{1}\left(A\right)\right)^{-1}$ where $\lambda_{1}\left(A\right)$ is the spectral radius of $A$, to define a vertex centrality index, nowadays known as Katz centrality.
	However, the resolvent of the adjacency matrix $\left(I-\varepsilon A\right)^{-1}$ is parametric, where the parameter $\varepsilon$ is upper-bounded by the reciprocal of $\lambda_{1}\left(A\right)$.
	Then when $\lambda_{1}\left(A\right)$ is significantly large, most of the information of the network structure stored in the $A$ matrix is making almost no contribution. This has been recently shown in examples where $\lambda_{1}\left(A\right)\gg 1$, and the resolvent of $A$ does not provide reasonable results \cite{estrada2024communicability}.
	Hence, the definitions of subgraph centrality $\left(e^{A}\right)_{uu}$ \cite{estrada2005subgraph} and communicability $\left(e^{A}\right)_{uv}$ \cite{estrada2008communicability} have triggered much recent interest.
	Another advance of the use of matrix exponential is its interpretation and derivation in different contexts, ranging from coupled quantum harmonic oscillators \cite{estrada2012physics} and compartmental epidemiological models \cite{lee2019transient}, to nonconservative
	diffusion \cite{estrada2024conservative}. 
	Last but not least, we can think of $e^{A}=\sum_{k=0}^{\infty}\left(k!\right)^{-1}A^{k}$ by replacing the factorial for its more general definition based on Euler Gamma functions: $E_{\alpha,\beta}\left(A\right)=\sum_{k=0}^{\infty}\left(\varGamma\left(\alpha k+\beta\right)\right)^{-1}A^{k}$, $\alpha,\beta>0$. It retrieves the exponential when $\alpha=1$ and $\beta=1$, but in general represents the Mittag-Leffler matrix functions of $A$. The idea of using $E_{\alpha,\beta}\left(A\right)$ to define centrality and communicability indices was previously developed independently
	by Arrigo and Durastante \cite{arrigo2021mittag} and by Estrada \cite{estrada2022many}.
	
	In this work, we take a step forward to using Mittag-Leffler matrix functions to quantify the degree of balance of signed graphs. 
	A signed graph $G_s$ can have both positive and negative edges
	\cite{zaslavsky1982signed}. 
	The signs of the edges emerge in various real-world scenarios. For instance, positive signs may represent friendship, collaboration, alliances, etc., while negative ones may represent enemity, hostility, conflicts, etc. in social networks \cite{Altafini_2012_opinion,isakov2019structure}. In voting systems, they may represent whether two voters support the same or different candidates ; in recommendation systems, they can correspond to whether two users recommend the same product, or they have discrepant opinions about the same product. In transcriptional networks, edges represent the action of a transcription factor on one of its target genes, and the sign means activation ($+$) or inhibition ($-$) \cite{soranzo2012decompositions}. Cooperation and competition between species in ecological networks \cite{saiz2017evidence} and between products in economic networks \cite{tian_2022_thesis,tian_2021_role} can also be assigned to positive and negative edges, respectively.
	
	The important notion of balance can be defined through the sign of cycles, which is the product of the signs of its edges \cite{cartwright1956structural,harary_1953_balance}. Specifically, a graph is balanced if and only if all its cycles are positive; otherwise it is unbalanced.
	If we focus on a signed triangle, it is balanced if either (i) all its edges are positive or (ii) two edges are positive and one is negative. The stability of the first triangle is self-evident, while in the second, the structure indicating that the ``enemy of my enemy is my friend'' provokes our feeling of stability by the formation of a coalition against the common enemy. 
	The all-negative triangle is clearly unbalanced, the same as the one with only one negative edge. In the latter case, there are clear tensions between the two vertices sharing the negative edge and the one with whom they share positive ones. Think about the tensions in a cycle of friends apart from one couple in conflict. We would expect that the tensions existing between the members decay as its length increases. 
	Hence, Estrada and Benzi has proposed the index $K\left(G_{s}\right)=tr\left[e^{A}\right]/tr\left[e^{\left|A\right|}\right]$ where $\left|A\right|$ is the entrywise absolute value of $A$, to quantify the degree of balance \cite{estrada2014walk}. In this way, $K\left(C_{3}^{-}\right)\approx0.686$, $K\left(C_{4}^{-}\right)\approx0.915$, and $K\left(C_{5}^{-}\right)\approx0.983$, where $C_k^-$ denotes the cycle of length $k$ and one negative edge. Further, $K\left(C_{10}^{-}\right)\approx0.99999947$, which is very close to balance (where $K=1$). Is it not the case that the factorial penalization used in the exponential is too heavy and fool us in this case? Here, by completing a close walk of length $10$, the information contained by the negative edge present in this cycle is almost completely forgotten.
	
	In this paper, we start by showing that the balance index $K\left(G_{s}\right)$ can be obtained from first principles on the basis of a nonconservative (NC) diffusion dynamic taking place on the graph $G_{s}$ relative to its underlying unsigned graph. Using this approach, we generalize the NC diffusion on graphs to a temporal-fractional model using Caputo fractional derivative. 
	In this way, we generalize the balance index $K(G_s)$ to indices based on Mittag-Leffler (ML) matrix functions of $A$. These new indices are derived from first-principles diffusion processes which are temporally non-local. Therefore, the ML balance index accounts for certain memory of the system about the past, by diminishing the penalization that long cycles typically receive. We illustrate our results with the use of some artificial signed graphs, as well as real-world networks representing gene transcription networks, ecological competition between plant species in vast regions of Spain, and social networks in rural villages in Honduras. 
	
	\section{Preliminaries}\label{sec:preliminary}
	Let us consider an undirected connected signed graph $G_s=\left(V,E,\rho\right)$ where $V=\{1,2,\dots, n\}$ is the vertex set, an edge $(i, j)\in E$ is an unordered pair of two distinct nodes in the set $V$, and $\rho: E\to\varSigma$, $\varSigma=\left\{ \pm1\right\}$, associates each edge with a sign. 
	Let $A\left(G\right)=A \in \mathbb{R}^{n\times n}$ be the adjacency matrix of $G$. Specifically, if there is no edge between nodes $i,j$, $A_{ij} = 0$; otherwise, $A_{ij}=\rho((i,j))$ denotes the edge sign. We will also consider the graph where we ignore the edge sign $\tilde{G}$, and the unsigned adjacency matrix $\left|A\right|$ where the absolute values are taken entrywise. 
	
	\subsection{Structural balance}
	A fundamental notion in the study of signed networks is the so-called \textit{structural balance} \cite{cartwright1956structural,heider_1946_psychology}. A signed graph is structurally balanced if and only if there is no cycle with an odd number of negative edges, which can be effectively defined through the following theorem.
	\begin{theorem}[structure theorem for balance \cite{harary_1953_balance}]
		A signed graph $G$ is structurally balanced if and only if there is a bipartition of the node set into $V=V_1\cup V_2$ with $V_1$ and $V_2$ being mutually disjoint and one of them being nonempty, s.t.~any edge between the two is negative while any edge within each node subset is positive. 
		\label{the:lit_signed-bal}
	\end{theorem}
	There are several indices proposed to quantify the degree of balance, e.g., \cite{estrada2019rethinking,estrada2014walk,Facchetti_2011_large,giscard2017index,harary_1953_balance,kirkly2019index,singh2017index,tian2024sign}. One of the first measures based on the walk lengths was proposed by Estrada and Benzi \cite{estrada2014walk},
	\begin{equation}
		K\left(G\right)\coloneqq\dfrac{Tr\left(e^{A}\right)}{Tr\left(e^{\left|A\right|}\right)} = \dfrac{\sum_{j=1}^ne^{\lambda_j}}{\sum_{j=1}^n e^{\mu_j}},
		\label{equ:K(G)}
	\end{equation}
	where $\lambda_j$ and $\mu_j$ denote the eigenvalues of $A$ and $\left|A\right|$, respectively. 
	
	We now introduce \textit{switching equivalence}, which generalizes the idea of balance. 
	\begin{definition}
		The operation of reversing the signs of all edges connecting a subset $S\subseteq V$ and its complement is called switching the subset $S$. Two signed configurations $\rho, \rho': E\to \varSigma$ are said to be switching equivalent if there exists $S\subseteq V$ such that $\rho'$ can be obtained from $\rho$ by switching the subset $S$, denoted by $\rho \approx \rho'$. 
	\end{definition}
	Switching equivalence is an equivalence relation on sign configurations of a fixed underlying graph, and the corresponding equivalent classes are called switching classes. Clearly, balanced graphs comprise one switching class. It is also known that the spectra of signed graphs are switching invariant \cite{Atay_signedCheeger_2020,zaslavsky1982signed}.
	
	\subsection{Laplacians and Mittag-Leffler matrix function}
	We consider the signed Laplacian $L_{A}$ as
	\begin{equation}
		L_{A}\left(i,j\right)=\left\{ \begin{array}{cc}
			\sum_{\left(i,j\right)\in E}\left|A_{ij}\right| & i=j\\
			-A_{ij} & i\neq j.
		\end{array}\right.
	\end{equation}
	It governs the diffusion dynamics by Altafini's
	consensus model \cite{altafini2013biconsensus} that we will talk about in more detail later.  
	We also introduce the Lerman-Ghosh Laplacian \cite{ghosh2024non,lerman2012network}, 
	\begin{align}
		L_{\chi} = \chi I-A,
		\label{equ:laplacian-LM}
	\end{align}
	where $\chi\in\mathbb{R}^{+}$, and $I$ is the identity matrix. 
	The index we will propose in this paper is closely related to the dynamics governed by this Laplacian, and we will show that it shares an important property with the dynamics by the signed Laplacian. 
	Specifically, we will apply the time-fractional Caputo derivative,
	\begin{align}
		D_{t}^{\alpha}u(t)=\frac{1}{\Gamma(1-\alpha)}\int_{0}^{t}\frac{u'(\tau)}{(t-\tau)^{\alpha}}d\tau ,
		\label{equ:caputo}
	\end{align}
	where $u'(\tau)$ denotes the usual derivative. We assume that $u$ is differentiable and the convolution can be defined. Here, $0<\alpha\leq 1$, $0<t<\infty$, and $\varGamma\left(x\right)$ is the Euler gamma function. 
	We also recall that a diffusion process is said to be conservative if the number of diffusive particles is constant along the time; otherwise, it is called a nonconservative diffusion \cite{estrada2024conservative}.
	Finally, we introduce the building block of the balance index we will propose, the Mittag-Leffler (ML) function of a matrix, say $M$,
	\begin{equation}	   
		E_{\alpha}\left(M\right)\coloneqq\sum_{k=0}^{\infty}\dfrac{M^{k}}{\varGamma\left(\alpha k+1\right)}.
		\label{equ:ML}
	\end{equation}
	The study of these matrix functions for networks was previously studied by Arrigo and Durastante \cite{arrigo2021mittag}, and they also proposed to use $E_{\alpha}^\gamma \coloneqq E_{\alpha}\left(\gamma M\right)$ with $\gamma \le \Gamma(\alpha+1)$ accounting for fair contribution of walks in graphs. In our implementations, we adopt this suggestion, with $\gamma = \Gamma(\alpha+1)$.

	\section{Motivation}
	\label{sec:motivation}
	There are $2^{15}$ ways to put signs on the edges of the Petersen graph, on which only five (excluding the unsigned one) are essentially different \cite{zaslavsky2012six} (see Fig.~\ref{Petersen graphs}). We consider the diffusion dynamics by Altafini's consensus model \cite{altafini2013biconsensus}.
	Let $u\left(t\right)$ be the vector representing the state of the vertices in $G_s$ at time $t$, with $u\left(0\right)=u^{0}$, and let $\dot{u}\left(t\right)$ be the vector of their time derivatives.
	Then,
	\begin{equation}
		\dot{u}_{i}\left(t\right)=-\sum_{\left(i,j\right)\in E}\left|A_{ij}\right|\left(u_{i}\left(t\right)-\textnormal{sgn}\left(A_{ij}\right)u_{j}\left(t\right)\right),\label{eq:diffusion}
	\end{equation}
	where $\text{sgn}(\cdot)$ returns the sign of the value. Hence,
	\begin{equation}
		\dot{u}\left(t\right)=-L_{A}u\left(t\right);u\left(0\right)=u^{0}.\label{eq:diffusion-1}
	\end{equation}
	\begin{center}
		\begin{figure}[h]
			\begin{centering}
				\includegraphics[width=0.7\textwidth]{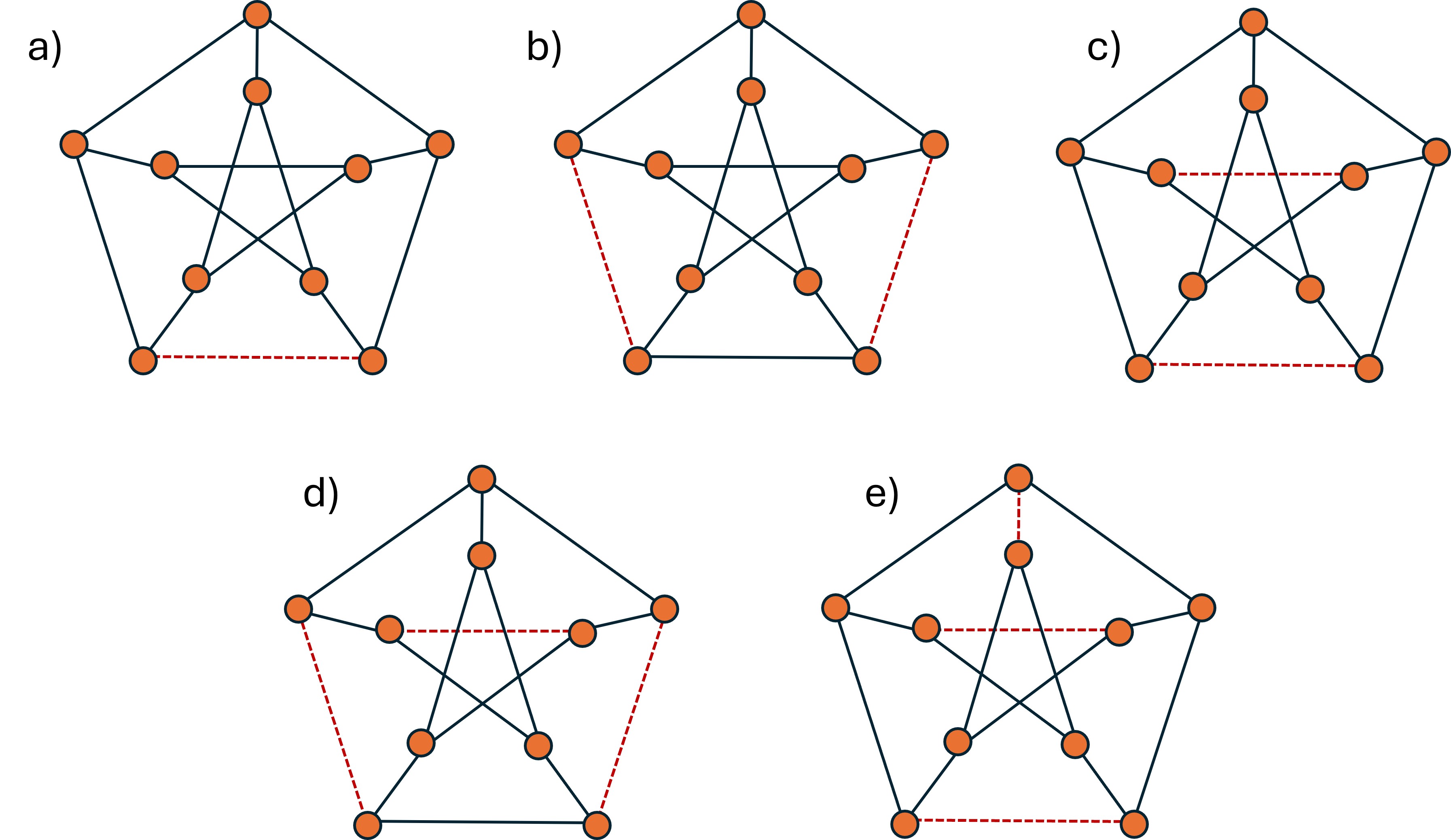} 
				\par\end{centering}
			\caption{The five switching isomorphism types of signed Petersen graph excluding the unsigned one, where solid lines represent positive edges and dashed lines represent negative ones.}
			\label{Petersen graphs} 
		\end{figure}
		\par\end{center}
	
	We consider the convergence time, $t_{c}$, at which the state values are sufficiently close to each other, with tolerance $10^{-5}$, i.e., $\left|u_{v}\left(t_{c}\right)-u_{u}\left(t_{c}\right)\right|<10^{-5}$, $\forall u,v\in G_s$. We note that the only difference of graphs in Fig.~\ref{Petersen graphs} lies in their sign patterns, and we denote a negative cycle of length $k$ by $C_{k}^{-}$. 
	We find that the graph having the most $C_{5}^{-}$, graph e), is the graph reaching the consensus in a fastest time, with 12 $C_{5}^{-}$ and $t_c=11$. The graph having the least $C_{5}^{-}$, graph a), is the one that delays the most, with 4 $C_{5}^{-}$ and $t_c=48$.
	It is known that consensus will never be reached if a graph is balanced, but instead the dynamics reaches a dissensus state. Therefore, graph a) is more similar to a balanced graph in the sense that it delays more to reach the consensus than graph e). 
	However, this simple arithmetic is broken when we consider that graphs b) and d), both with 6 $C_{5}^{-}$, but with the first almost doubling the time for consensus of the second ($24$ versus $14$). We can then extend the analysis to consider $C_{6}^{-}$, which clearly indicates that graph d) is less similar to a balanced graph than graph b), with $10$ versus $6$ $C_{6}^{-}$, respectively. Under this kind of semiquantitative analysis, a problem emerges when considering graph c) with $t_c=22$, which has $8$ $C_{5}^{-}$, more than that of graphs b) and d), but 4 $C_{6}^{-}$ less than that of the previous two graphs. We defer more details to Supplementary Material.
	
	Since the Petersen graphs are cubic, $L_{A}\left(i,i\right)=3$, $\forall i\in V$, which allows us to write $L_{A}=L_{\chi}=\chi I-A$, where
	$L_{\chi}$ is the Lerman-Ghosh Laplacian \eqref{equ:laplacian-LM}
	as introduced in section \ref{sec:preliminary}, and $\chi=3$ here. The solution to the Cauchy problem (\ref{eq:diffusion}) is
	\begin{equation}
		u\left(t\right)=e^{-tL_{\chi}}u^{0}=e^{-t\chi}e^{tA}u^{0}.
	\end{equation}
	Then, at a given time $t$, the concentration at a vertex $v$ is
	\begin{equation}
		u_{v}\left(t\right)=\sum_{j}\left(e^{-t\chi}e^{tA}\right)_{v,j}u_{j}^{0}.
	\end{equation}
	Suppose that the initial concentration is totally located
	at the vertex $v$, $u_{j}^{0}=\delta_{j,v}$ , where	$\delta_{i,j}$ is the Kronecker delta, then
	\begin{equation}
		u_{v}\left(t\right)=e^{-t\chi}\left(e^{tA}\right)_{vv}.
	\end{equation}
	Then the total concentration remaining at the vertices
	when the initial concentration has been totally allocated at them is
	\begin{equation}
		T_{s}\coloneqq\sum_{v}u_{v}\left(t\right)=e^{-t\chi}\sum_{v}\left(e^{tA}\right)_{vv}=e^{-t\chi}Tr\left(e^{tA}\right),
	\end{equation}
	where $Tr\left(\cdot\right)$ returns the trace of a matrix.
	In a similar way, we can ignore the edge sign and consider the underlying graph $\tilde{G}$,
	\begin{equation}
		T_{u}\coloneqq\sum_{v}u_{v}\left(t\right)=e^{-t\chi}\sum_{v}\left(e^{t\left|A\right|}\right)_{vv}=e^{-t\chi}Tr\left(e^{t\left|A\right|}\right).
	\end{equation}
	A way to account for the ``influence'' of the edge signs on the diffusion is $T_{s}/T_{u}$, such that for $t=1$ we recover the measure $K(G)$ in \eqref{equ:K(G)}. 
	This balance index can be easily generalized by taking any value of $t=\beta$, $K(G,\beta)=Tr(\exp(\beta A))/Tr(\exp(\beta \left|A\right|))$. 
	
	For the five nonsimilar signed Petersen graphs, although there is a good correlation between $K\left(G\right)$ and $t_{c}$ (Pearson correlation: $r^{2}\approx0.924$), there is an important inversion in the values of $K\left(G\right)$ for graphs c) and d).
	Specifically, for $K(G)$, graph c) has value $0.941$ while graph d) has value $0.947$, but graph c) has larger $t_c$ than graph d); see Supplementary Materials for details. 
	The problem seems to be produced by the differences in the penalization that the cycles of length 5 and 6 receive in the exponential function. To see this,
	we examine the difference between $Tr\left(A^{k}\right)/k!$
	for graph c) and graph d) for values of $1\leq k\leq10$; see Fig.~\ref{Contributions}. We find that the largest contribution is the one of $k=5,$ which is about $-0.333$, followed by that of $k=6,$ which is $0.2$. This reflects the fact that graph c) has more negative cycles of length 5 than d), that d) has more negative cycles of length 6 than c), but that cycles of length 6 are much heavily penalized than those of length 5. 
	We can put it in the following way. If one has to pay \$1 for every $C_{5}^{-}$ but only \$0.1 for each $C_{6}^{-}$, graph c) will have to pay \$8.40, while only \$7.00 is needed for graph d).
	But if the penalty for $C_{6}^{-}$ increases to \$0.5, then graph c) will need to pay \$10 while \$11 will be paid by graph d).
	Therefore, the problem we raise in this paper is how to tune the penalization of longer cycles, such that their contribution to the balance/unbalance of networks becomes more relevant when necessary. We propose to achieve it while keeping the first principles explained before that connect the balance index $K\left(G\right)$ with a (nonconservative) diffusion on graphs. 
	\begin{center}
		\begin{figure}[h]
			\begin{centering}
				\includegraphics[width=0.42\textwidth]{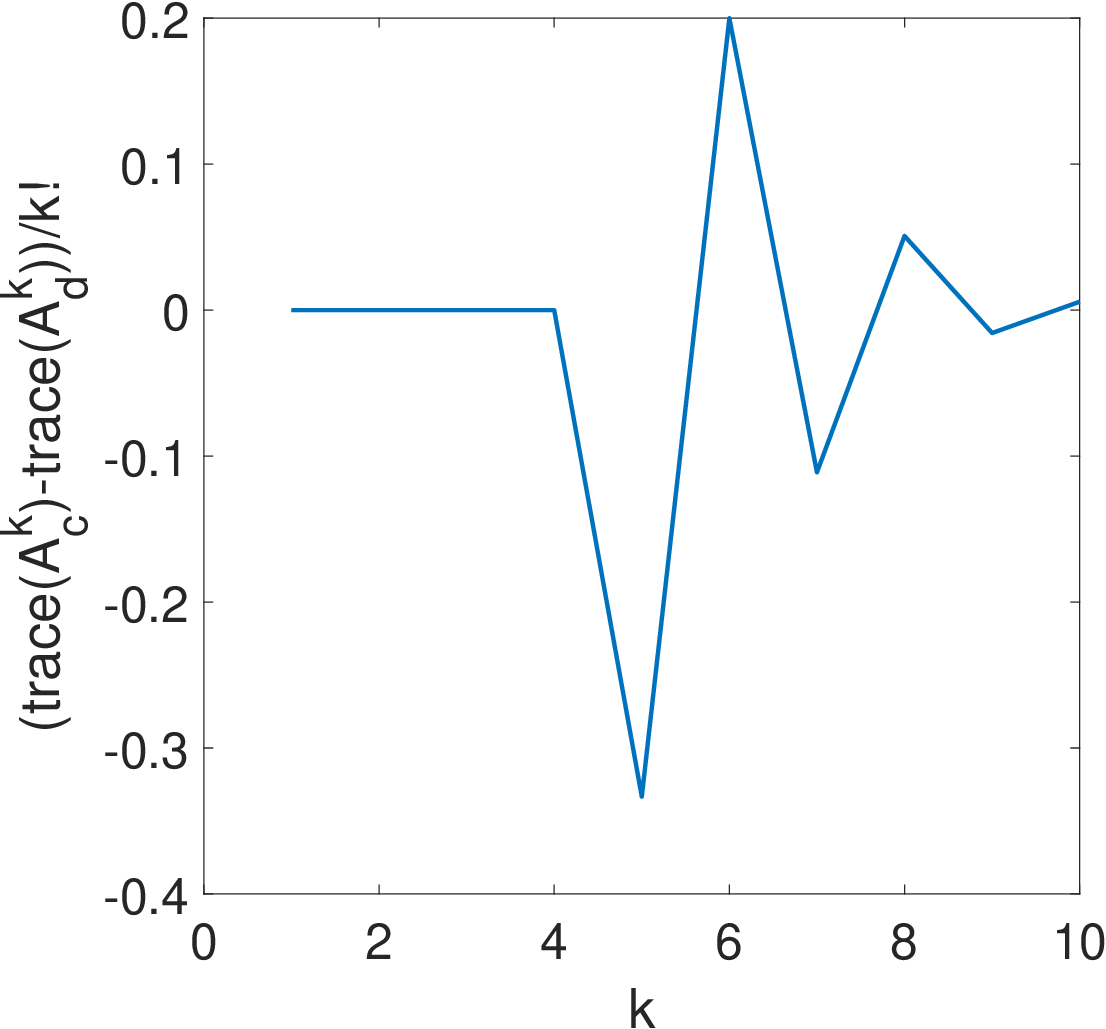} 
				\par\end{centering}
			\caption{Change of the difference of $Tr\left(A^{k}\right)/k!$ between graphs
				c) and d) of Fig. \ref{Petersen graphs}.
			}
			
			\label{Contributions} 
		\end{figure}
		\par\end{center}
	
	We end up this section by proving that Altafini's model of consensus on signed graphs is nonconeservative, unless the graph is balanced and the initial vector is the eigenvector corresponding to the eigenvalue $0$. We should also notice that when the graph does not contain any negative edge, Altafini's consensus model is effectively the consensus model with the graph Laplacian and it is conservative. 
	\begin{proposition}
		The diffusion by Altafini's consensus model is nonconservative, unless
		the graph is balanced and the intial vector $u^{0}$ is the eigenvector
		corresponding to the smallest eigenvalue $0$. 
	\end{proposition}
	\begin{proof}
		The solution to the Altafini's consensus is 
		\begin{align}
			u(t)=e^{-tL_{A}}u^{0}.
		\end{align}
		Let $0\le\mu_{1}\le\mu_{2}\le\dots\le\mu_{n}$ be the eigenvalues
		of $L_{A}$, and let $\phi_{i}$ be the orthonormal eigenvector associated
		with $\mu_{i}$. Then, 
		\begin{align}
			u(t)=e^{-t\mu_{1}}\phi_{1}\phi_{1}^{T}u^{0}+e^{-t\mu_{2}}\phi_{2}\phi_{2}^{T}u^{0}+\dots+e^{-t\mu_{n}}\phi_{n}\phi_{n}^{T}u^{0}.\label{equ:u-eig}
		\end{align}
		
		We know that if a signed graph is unbalanced, $\mu_{1}>0$. Then,
		\begin{align*}
			\lim_{t\to\infty}u(t)=\mathbf{0},
		\end{align*}
		where $\mathbf{0}$ is the all-zero vector. Hence, the diffusion is
		nonconservative.
		
		We now consider the balanced case, where $\mu_{1}=0$ and $\mu_{2}>0$.
		Then, 
		\begin{align*}
			\lim_{t\to\infty}u(t)=\phi_{1}\phi_{1}^{T}u^{0}.
		\end{align*}
		Hence $\lim_{t\to\infty}\mathbf{1}^{T}u(t)=\mathbf{1}^{T}u^{0}$ if
		and only if $u^{0}=\phi_{1}$. In the case of $u^{0}=\phi_{1}$, from
		Eq.~\eqref{equ:u-eig}, we have $\mathbf{1}^{T}u(t)=\mathbf{1}^{T}u^{0}$,
		by the orthogonality of eigenvectors. Hence, the diffusion is conservative
		if and only if the graph is balanced and $u^{0}=\phi_{1}$. 
	\end{proof}
	
	
	\section{Main results}
	\label{sec:main}
	
	\subsection{Nonconservative fractional diffusion and balance}
	We know that Altafini's dynamics on signed graphs is nonconservative. Let us now consider a more general nonconservative diffusive model on the signed graph based on the Lerman-Ghosh Laplacian $L_\chi$ \cite{ghosh2024non,lerman2012network}. To make the process more general, we also replace the standard time derivative $\dot{u}\left(t\right)$ by the time-fractional Caputo derivative $D_{t}^{\alpha}$ as in Eq.~\eqref{equ:caputo}.
	Hence, the nonconservative diffusion on the signed graph we consider is
	\begin{equation}
		D_{t}^{\alpha}u\left(t\right)=-L_{\chi}u\left(t\right);u\left(0\right)=u^{0}.
		\label{eq:diffusion-1-1}
	\end{equation}
	The solution of Eq. (\ref{eq:diffusion-1-1}) is given by 
	\begin{align}
		u(t)=E_{\alpha}(-t^{\alpha}L_{\chi})u^{0},\label{matrix_solution}
	\end{align}
	where $E_{\alpha}(\cdot)$ is the Mittag-Leffler function as in Eq.~\eqref{equ:ML}.
	Let us focus again on the concentration at a vertex designated by $v$,
	\begin{equation}
		u_{v}\left(t\right)=\sum_{j}\left(E_{\alpha}\left(-t^{\alpha}L_{\chi}\right)\right)_{vj}u_{j}^{0},
	\end{equation}
	and if the initial concentration is totally located at the vertex $v$, $u_{j}^{0}=\delta_{j,v}$, we get
	\begin{equation}
		u_{v}\left(t\right)=\left(E_{\alpha}\left(-t^{\alpha}L_{\chi}\right)\right)_{vv}.
	\end{equation}
	One main difference between the exponential and the Mittag-Leffler function is that in general $E_{\alpha}\left(P+Q\right)\neq E_{\alpha}\left(P\right)E_{\alpha}\left(Q\right)$, even when $P$ and $Q$ commute \cite{sadeghi2018some}. This equality holds in general only when (i) $P$ and $Q$ commute and (ii) $\alpha=1$. 
	
	Here, we consider the special case when $\chi=0$, such that
	\begin{align}
		D_{t}^{\alpha}u\left(t\right)=Au\left(t\right);u\left(0\right)=u^{0}.
		\label{eq:diffusion-1-chi0}
	\end{align}
	The concentration at vertex $v$ with $u_{j}^{0}=\delta_{j,v}$ is
	\begin{align*}
		u_{v}\left(t\right)=\left(E_{\alpha}\left(t^{\alpha}A\right)\right)_{vv}.
	\end{align*}
	Then the total concentration remaining at the vertices when the initial concentration has been totally allocated at them is
	\begin{equation}
		\tilde{T}_{s}\coloneqq\sum_{v}u_{v}\left(t\right)=\sum_{v}\left(E_{\alpha}\left(t^{\alpha}A\right)\right)_{vv}=Tr\left(E_{\alpha}\left(t^{\alpha}A\right)\right).
	\end{equation}
	Similarly, we can ignore the edge sign and obtain the total concentration in $\tilde{G}$, 
	\begin{equation}
		\tilde{T}_{u}\coloneqq\sum_{v}u_{v}\left(t\right)=\sum_{v}\left(E_{\alpha}\left(t^{\alpha}\left|A\right|\right)\right)_{vv}=Tr\left(E_{\alpha}\left(t^{\alpha}\left|A\right|\right)\right).
	\end{equation}
	Finally, we summarise the influence of the edge signs on the diffusion as the ratio $\tilde{T}_{s}/\tilde{T}_{u}$,
	such that for $t=1$ we have
	\begin{equation}
		K_{\alpha}\left(G_{s}\right)\coloneqq\dfrac{Tr\left(E_{\alpha}\left(A\right)\right)}{Tr\left(E_{\alpha}\left(\left|A\right|\right)\right)}.
	\end{equation}
	We note that $K\left(G\right)$ is the particular case when $\alpha=1$.
	
	\subsection{How global balance is accounted for}
	A closed walk (CW) is said to be positive (negative) if the product of the signs of all its composing edges is positive (negative). Let $M_{k}\left(i,i\right)=\left(A^{k}\right)_{ii}$ be the total ``number" of CWs of length $k$ starting at vertex $i$, then
	\begin{equation}
		M_{k}\left(i,i\right)=\mu_{k}^{+}\left(i,i\right)-\mu_{k}^{-}\left(i,i\right),
	\end{equation}
	where $\mu_{k}^{+}\left(i,i\right)$ is the number of positive CWs of length $k$ starting at $i$, and $\mu_{k}^{-}\left(i,i\right)$ is the same for negative CWs \cite{diaz2024signed}. Obviously,
	\begin{equation}
		\begin{split}
			Tr\left(E_{\alpha}\left(A\right)\right) 
			=\sum_{i=1}^{n}\sum_{k=0}^{\infty}\dfrac{M_{k}\left(i,i\right)}{\varGamma\left(\alpha k+1\right)}
			=Tr\left(E_{\alpha}^{+}\left(A\right)\right)-
			Tr\left(E_{\alpha}^{-}\left(A\right)\right),
		\end{split}
	\end{equation}
	where $Tr\left(E_{\alpha}^{\pm}\left(A\right)\right)=\sum_{i=1}^{n}\sum_{k=0}^{\infty}\mu_{k}^{\pm}\left(i,i\right)/\varGamma\left(\alpha k+1\right)$ are the positive and negative contributions to $Tr\left(E_{\alpha}\left(A\right)\right)$. We note that they are not the same as $Tr\left(E_{\alpha}\left(A^{\pm}\right)\right)$ where $A^{\pm}$ are the adjacency matrices only for positive and negative edges of $G_{s}$, respectively. Similarly, 
	\begin{equation}
		\begin{split}Tr\left(E_{\alpha}\left(\left|A\right|\right)\right) & =Tr\left(E_{\alpha}^{+}\left(A\right)\right)+
			Tr\left(E_{\alpha}^{-}\left(A\right)\right)\end{split}.
	\end{equation}
	Hence, 
	\begin{equation}
		K_{\alpha}\left(G_{s}\right)=\dfrac{Tr\left(E_{\alpha}^{+}\left(A\right)\right)-
			Tr\left(E_{\alpha}^{-}\left(A\right)\right)}{Tr\left(E_{\alpha}^{+}\left(A\right)\right)+
			Tr\left(E_{\alpha}^{-}\left(A\right)\right)}.
		\label{eq:Ka-pn}
	\end{equation}
	
	We now understand $K_\alpha(G_s)$ through its two different terms. 
	Let us recall that a trivial CW is a walk starting at and ending at the same vertex but not involving any cycle in the graph. Hence, any trivial CW is always positive. Therefore, $Tr\left(E_{\alpha}^{+}\left(A\right)\right)$ accounts for all trivial CWs and nontrivial positive CWs.
	In a nontrivial positive CW, there can be any number of balanced cycles, and also even number of unbalanced cycles. 
	We can understand it as follows. Consider a negative triangle with sign pattern $+,+,-$ for edges $(A,B), (B,C), (A,C)$, respectively. Then, a voting system on this triangle will end up in contradictions after one round of information passing. For example, if A votes Y(N), then B will vote Y(N), and C will also vote Y(N), but then A will need to vote N(Y) since it is in conflict with C, contradicting its initial vote.
	However, if the number of rounds is even such contradictions disappear, eliminating the tension in the system. 
	In closing, the term $Tr\left(E_{\alpha}^{+}\left(A\right)\right)$ accounts for all CWs in the signed graph that involves no tensions from the perspective of balance. 
	This necessarily leads to the fact that all tensions are encoded in $Tr\left(E_{\alpha}^{-}\left(A\right)\right)$. Indeed, any negative CW necessarily contains a negative cycle, which by definition is unbalanced. Therefore, the difference $Tr\left(E_{\alpha}^{+}\left(A\right)\right)-
	Tr\left(E_{\alpha}^{-}\left(A\right)\right)$ accounts for the magnitude of ``tensions'' existing in the signed graph in terms of balance, such that $K_{\alpha}\left(G_{s}\right)$ will be $0$ if the balanced and unbalanced contributions are equal,
	and will be $1$ if there are no unbalanced contributions.
	
	\subsection{How memory is accounted for}
	We first show that the time-fractional Caputo derivative accounts for the memory of the system about its past. We start by writing
	\begin{equation}
		\begin{split}D_{t}^{\alpha}u\left(t\right)  =\dfrac{1}{\varGamma\left(1-\alpha\right)}\int_{0}^{t}\left[\dfrac{1}{\left(t-\tau\right)^{\alpha}}\right]u'\left(\tau\right)d\tau
			=  \dfrac{1}{\varGamma\left(1-\alpha\right)}\int_{0}^{t}w\left(\tau\right)u'\left(\tau\right)d\tau,
		\end{split}
	\end{equation}
	where $w\left(\tau\right) = 1/\left(t-\tau\right)^{\alpha}$ is used to indicate that
	$u'\left(\tau\right)$ is integrated in a weighted way that $w\left(\tau\rightarrow0\right)$
	is significantly smaller than $w\left(\tau\rightarrow t\right)$.
	Odibat \cite{odibat2006approximations} has proved that
	\begin{equation}
		D_{t}^{\alpha}u\left(t\right)=C\left(u,h,\alpha\right) - E_{C}\left(u,h,\alpha\right),
	\end{equation}
	where $E_{C}\left(u,h,\alpha\right)\leq\mathcal{O}\left(h^{2}\right)$ is the error term, and
	\begin{equation}
		\begin{split}C\left(u,h,\alpha\right) & =\dfrac{h^{1-\alpha}}{\varGamma\left(3-\alpha\right)}\left[\underset{\textnormal{remote past}}{\underbrace{\left(\left(k-1\right)^{2-\alpha}-\left(k+\alpha-2\right)k^{1-\alpha}\right)u'\left(0\right)}}\right.\\
			& \left.+\underset{\textnormal{recent past}}{\underbrace{\sum_{j=1}^{k-1}\left(\left(k-j+1\right)^{2-\alpha}-2\left(k-j\right)^{2-\alpha}+\left(k-j-1\right)^{2-\alpha}\right)u'\left(t_{j}\right)}} +\underset{\textnormal{present}}{\underbrace{u'\left(t\right)}}\right],
		\end{split}
		\label{eq:trapezoidal}
	\end{equation}
	where the interval $\left[0,t\right]$ has been subdivided into $k$ subintervals $\left[t_{j},t_{j+1}\right]$ for $j=0,\ldots,k$ of equal length $h=t/k$. The term $C\left(u,h,\alpha\right)$ confirms that differently from the standard time derivative which considers only the present, the Caputo one takes into account the ``remote past'' and ``recent past'' together with the ``present'' state of the evolution of the function $u\left(\cdot\right)$. 
	Additionally, the time-fractional Caputo derivative gives smaller weight to the remote past, and such weight increases as we approach to the contribution of the present, which receives the largest weight.
	
	Let us now see the special case of $\alpha=1$ and how memory could be incorporated while changing $\alpha$. 
	We note that $C\left(u,h,\alpha=1\right)=u'\left(t\right)$.
	The solution of the NC diffusion \eqref{eq:diffusion-1-chi0} with $u_{j}^{0}=\delta_{j,v}$ is given by $u_{v}\left(t\right)=\left(e^{tA}\right)_{vv}$,
	i.e., the exponential of the adjacency matrix. For $t=1$, we know that 
	\begin{equation}
		e^{A}=I+A+\dfrac{A^{2}}{2!}+\ldots+\dfrac{A^{k}}{k!}+\ldots,
		\label{eq:expansion}
	\end{equation}
	which means that walks taking a large number of steps are so heavily penalized by $1/\left(k!\right)$ that they are almost forgotten. 
	Let us consider again an unbalanced cycle of $10$ vertices and only one negative edge, $C_{10}^-$. Here, we truncate the expressions (\ref{eq:expansion}) and $e^{\left|A\right|}$ at a given value $k$, denoted by $e^{A}\left(k\right)$ and $e^{\left|A\right|}\left(k\right)$, respectively.
	Then, for any $k<10$, we have that $Tr(e^{A}\left(k\right))=Tr(e^{\left|A\right|}\left(k\right)).$ Therefore, any penalization $c_{k}$ in $f\left(A\right)=\sum_{k=0}^{\infty}c_{k}A^{k}$ that makes $c_{10}Tr(\left|A\right|^{10})\approx 0$ will lead to $Tr\left(f\left(A\left(C_{10}^-\right)\right)\right) \approx Tr\left(f\left(\left|A\left(C_{10}^-\right)\right|\right)\right)$. This is exactly what happens with $c_{k}=1/k!$, where $Tr\left(\exp\left(A\left(C_{10}^-\right)\right)\right) \approx 22.7958$ and $Tr\left(\exp\left(\left|A\left(C_{10}^-\right)\right|\right)\right)\approx22.7959$ leading to $K_{1}\left(C_{10}^-\right)\approx0.99999947$. That is, the index has almost completely forgotten that the graph contains a negative edge. 
	However, the extra freedom introduced in $\alpha$ allows us to incorporate the information in the past in an appropriate manner. 
	For $C_{10}^-$, $\alpha=0.5$ makes that even the remote past receives some weight in the navigation of the diffusive particles, remembering the presence of the negative edge, with $K_{0.5}\approx 0.98290619$. Such memory can further take effect by dropping $\alpha$, which may be considered as the memory effect parameter, e.g., $K_{0.25}\approx0.109$.
	
	We now proceed to find the analytical expression for the degree of balance with memory for unbalanced cycles, i.e., cycles with an odd number of negative edges. We start by proving that unbalanced cycles share the same spectrum, independently of the exact number of negative edges. 
	\begin{proposition}
		There are two switching classes for signed cycles of length $n$, one corresponding to balance and the other corresponding to unbalance.
	\end{proposition}
	
	\begin{proof}
		For balanced cycles of length $n$, we know that they form a switching class. We then prove that all unbalanced signed cycles form one switching class. For an unbalanced signed cycle of length $n$, denoted by $C_{n}$, if we randomly remove one edge $e$, it becomes a tree $T_{n}$. We know that every signed tree is balanced, hence $T_{n}$ is balanced and is switching equivalent to the all-positive configuration. Edge $e$ necessarily breaks the balance structure, since otherwise $C_n$ is balanced. Hence, only one edge violates the balance structure, and $C_{n}$ is then switching equivalent to the unbalanced cycles of length $n$ and one negative edge. Hence, all unbalanced signed cycles of length $n$
		form one switching class. 
	\end{proof}
	\begin{corollary}
		All unbalanced signed cycles of length $n$ share the same eigenvalues, i.e., they are cospectral.
	\end{corollary}
	
	As a consequence of the previous results we will focus on the analytical study of unbalanced cycles as a general class.
	\begin{definition}[\cite{martin2023fractional}]
		Let $E_{\alpha}\left(z\right)$ be the Mittag-Leffler function of $z.$ Then, we define the following integral:
		\begin{equation}
			\mathcal{E}_{\nu,\alpha}\left(z\right)\coloneqq\frac{1}{\pi}\int_{0}^{\pi}\cos\left(\nu\theta\right)E_{\alpha}\left(z\cos\theta\right)d\theta,\:\nu\in\mathbb{Z}.
			\label{eq:Bessel}
		\end{equation}
	\end{definition}
	
	\begin{remark}
		Notice that $\mathcal{E}_{\nu,\alpha=1}\left(z\right)=\frac{1}{\pi}\int_{0}^{\pi}e^{z\cos\theta}\cos\left(\nu\theta\right)d\theta\eqqcolon I_{\nu}\left(z\right)$ is the modified Bessel function of the first kind. The fractional modified Bessel function of the first kind can be calculated by using the following result.
	\end{remark}
	
	\begin{lemma}[\cite{martin2023fractional}]
		Let $\mathcal{E}_{\nu,\alpha}(z)$ be the fractional modified Bessel function of the first kind of $z$ with fractional parameter $\alpha$ and $\nu\in\mathbb{Z}$. Then, 
		
		\begin{equation}
			\mathcal{E}_{\nu,\alpha}(z)=\sum_{k=0}^{\infty}\frac{\left(2k+\nu\right)!}{\varGamma\left(\alpha\left(2k+\nu\right)+1\right)k!\left(k+\nu\right)!}\left(\frac{z}{2}\right)^{2k+\nu}.
		\end{equation}
	\end{lemma}
	\begin{theorem}
		Let $C_{n}^{-}$ be the cycle graph with an odd number of negative edges and $0<\alpha\leq1$. Then,
		\begin{equation}
			K_{\alpha}^\gamma\left(C_{n}^{-}\right)=\dfrac{\sum_{k=1}^{n}E_{\alpha}\left(2\gamma\cos\left(\dfrac{\left(\left(2k+1\right)\pi\right)}{n}\right)\right)}{\sum_{k=1}^{n}E_{\alpha}\left(2\gamma\cos\left(\dfrac{\left(2k\pi\right)}{n}\right)\right)},
			\label{eq:cycle_Ka}
		\end{equation}
		where $K_{\alpha}^\gamma \coloneqq Tr\left(E_\alpha\left(\gamma A\right)\right) / Tr\left(E_\alpha\left(\gamma \left|A\right|\right)\right)$ is a even more general form of the index $K_{\alpha}$, and  $\gamma$ is a positive constant, 
		and
		\begin{equation}
			\underset{n\rightarrow\infty}{\lim}\frac{1}{n}\sum_{k=1}^{n}E_{\alpha}\left(2\gamma\cos\left(\dfrac{\left(2k\pi\right)}{n}\right)\right) = \mathcal{E}_{0,\alpha}\left(2\gamma\right).\label{eq:cycle_ML}
		\end{equation}
	\end{theorem}
	\begin{proof}
		We can obtain Eq.~\eqref{eq:cycle_Ka} by the eigenvalues of cycles of size $n$ and those of the same size and one negative edge \cite{mithai2012cycle}. 
		Then for the limit, since for $k=1,2,\dots,n$, the angles $2k\pi/n$ uniformly cover the interval $[0,2\pi]$, we can write
		\begin{equation}
			\underset{n\rightarrow\infty}{\lim}\left(\dfrac{1}{n}\sum_{k=1}^{n}E_{\alpha}\left(2\gamma\cos\left(\dfrac{2k\pi}{n}\right)\right)\right)
			= \left(\frac{1}{2\pi}\int_{0}^{2\pi}E_{\alpha}\left(2\gamma\cos\vartheta\right)d\vartheta\right),
		\end{equation}
		where
		\begin{equation}
			\left(\frac{1}{2\pi}\int_{0}^{2\pi}E_{\alpha}\left(2\gamma\cos\vartheta\right)d\vartheta\right)
			= \left(\frac{1}{\pi}\int_{0}^{\pi}E_{\alpha}\left(2\gamma\cos\vartheta\right)d\vartheta\right)
			= \mathcal{E}_{\alpha,0}\left(2\gamma\right).
		\end{equation}
	\end{proof}
	We cannot apply the same approximation as in the proof of (\ref{eq:cycle_ML}) to the numerator of $K_\alpha^\gamma$, because approximating 
	the numerator to
	the denominator for very large $n$ largely depends on the values of $\alpha$. 
	For instance, for $\alpha=1$ when $n=10$ the difference between the two terms is of the order of $10^{-6}$ and it drops to $10^{-15}$ for $n=20.$ However, for $\alpha=0.25$, this difference is of the order of $10^{7}$ for $n=10$ and remains of the order of $10^{6}$ for $n=20,$ and of $10^{3}$ for $n=40.$ 
	Therefore, because 
	the denominator can be approximated by $\mathcal{E}_{\alpha,0}\left(2\gamma\right)$,
	we have that for $\alpha=1$, the balance index approaches $1$ for relatively small unbalanced cycles, while this is far from being the case for
	$\alpha=0.25$. This is visually clear when we examine the change of $K_{\alpha} \left(C_{n}\right)$
	as a function of both $n$ and $\alpha$ (note that $\gamma = \varGamma\left(\alpha+1\right)$ throughout the paper so we ignore the superscript); see Fig.~\ref{cycles} (left). Specifically, for values of $\alpha$ close to $1$, the values of $K_{\alpha}\left(C_{n}\right)$ are close to $1$ for almost all cycles with size $n\geq 10$. At the other
	extreme when $\alpha$ is close to $0$, the values of the balance index are extremely low for almost every cycle with $n \leq20$. 
	
	\begin{figure}[htbp]
		\centering
		\begin{tabular}{cc}
			\includegraphics[height=0.25\textheight]{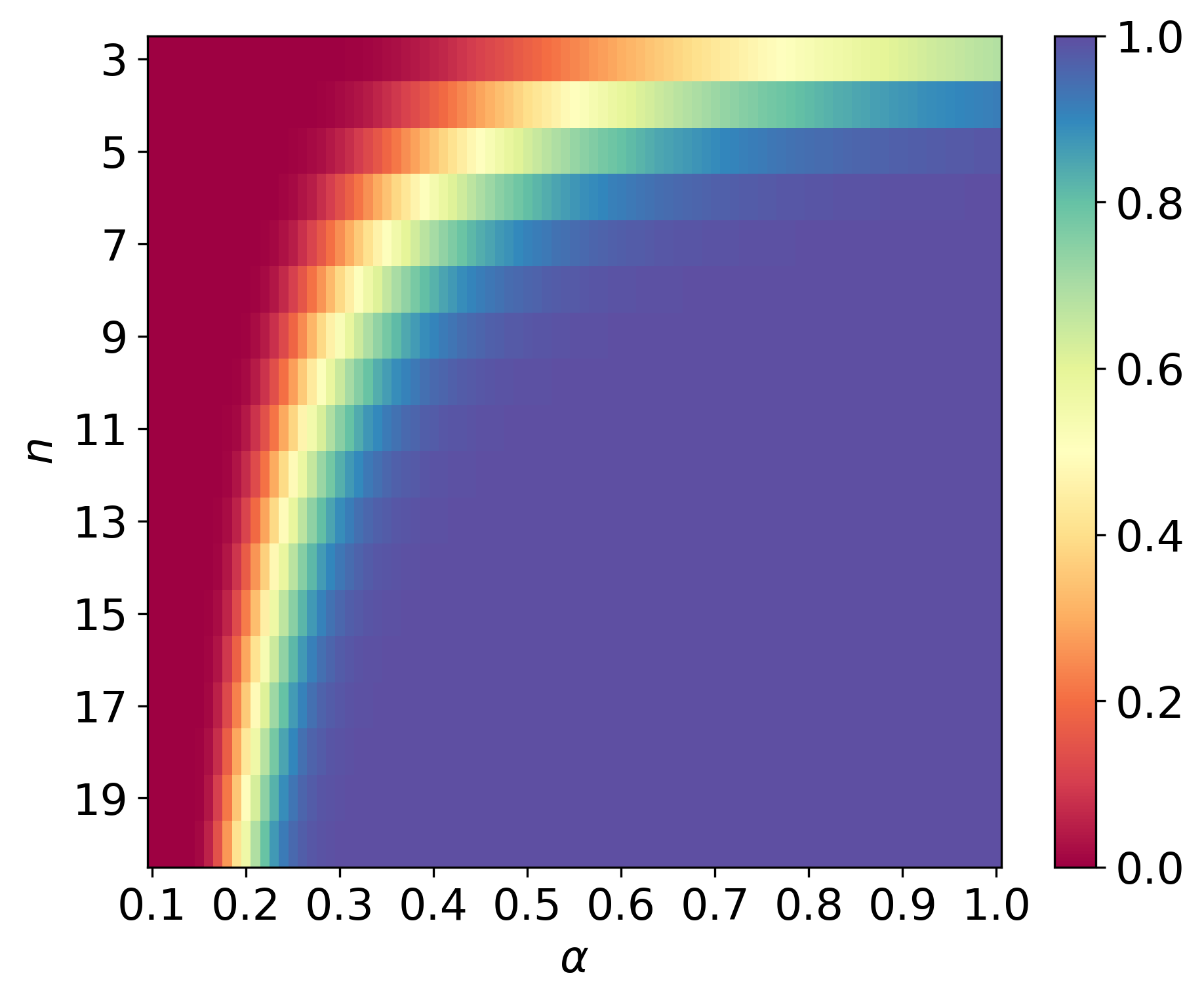} & \includegraphics[height=0.24\textheight]{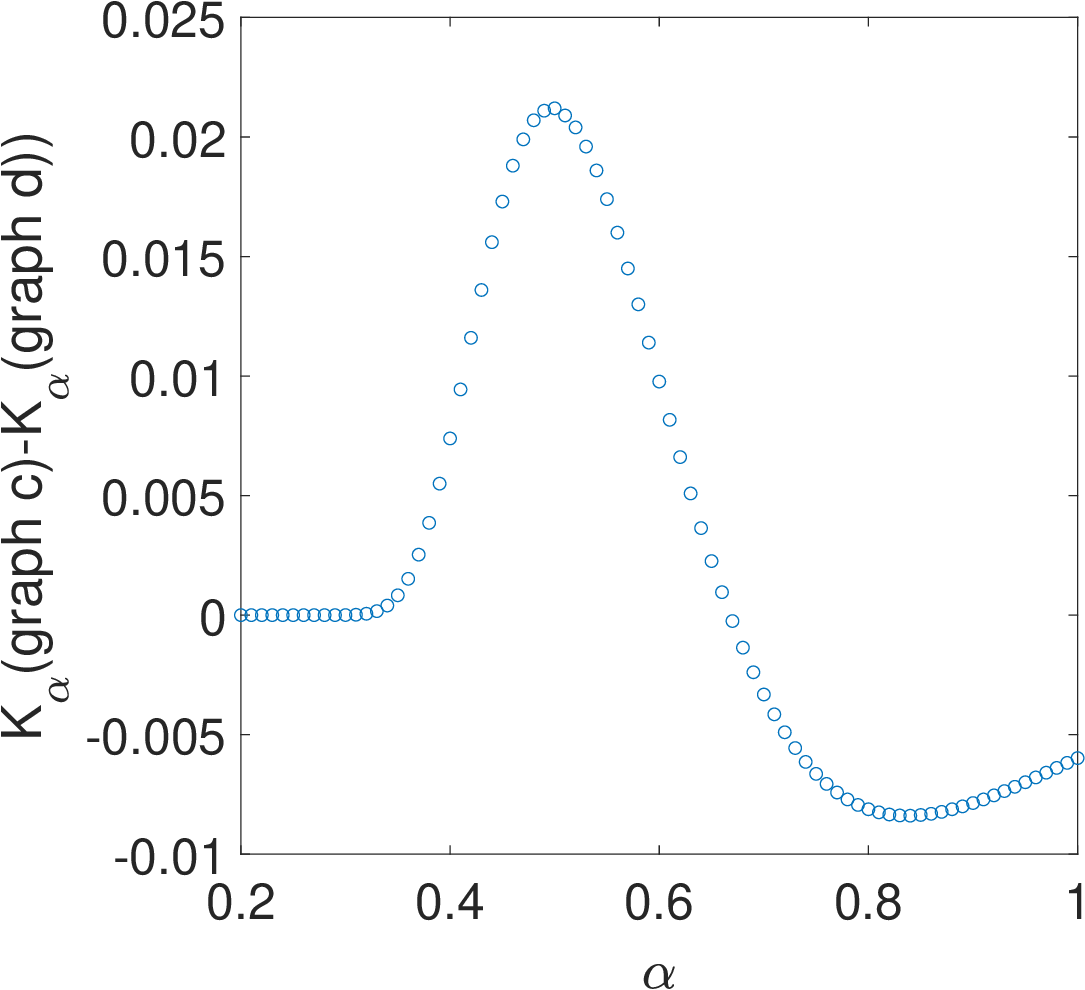}
		\end{tabular}
		\caption{(left) Contour plot of the values of $K_{\alpha}\left(C_{n}\right)$ for unbalanced cycles with $n$ vertices for values of $0.1\protect\leq\alpha\protect\leq1.$ (right) Difference of the balance index $K_\alpha$ for graph c) versus that for graph d) in Fig.~\ref{Petersen graphs} for values of $0.2\leq \alpha \leq 1$.}
		
		\label{cycles}
	\end{figure}

	\subsection{Properties of the balance index with memory}
	We start with the range of the balance index we have proposed. 
	\begin{theorem}
		The index $K_{\alpha}\left(G\right)$ is bounded as 
		\begin{equation}
			0\leq K_{\alpha}\left(G\right)\leq1,
		\end{equation}
		where the upper bound is reached if and only if the signed graph $G$ is balanced. 
	\end{theorem}
	\begin{proof}
		It is clear from Eq.~\eqref{eq:Ka-pn} that $K_\alpha(G) \le 1$. We now examine the lower bound.
		Let $\left\{ \lambda_{j}\left(A\right)\right\} =\left\{ \lambda_{j}^{+}\left(A\right)\right\} \cup\left\{ \lambda_{j}^{-}\left(A\right)\right\}$ be the eigenvalues of $A$, where $\left\{ \lambda_{j}^{+}\left(A\right)\right\}$ and $\left\{ \lambda_{j}^{-}\left(A\right)\right\}$ are the sets of nonnegative and negative eigenvalues of $A$, respectively. 
		Clearly, $E_{\alpha}\left(\lambda_{j}^{+}\left(A\right)\right)\geq 0$.
		We then consider negative eigenvalues $\lambda_{j}^{-}\left(A\right)$, and $E_{\alpha}\left(-\left|\lambda_{j}^{-}\left(A\right)\right|\right)$. 
		We note that
		\begin{equation}
			\left(-1\right)^{k}\dfrac{d^{k}E_{\alpha}\left(-x\right)}{dx^{k}}\geq0,
		\end{equation}
		for all $x$ and for $0\leq\alpha\leq1$ \cite{polard1948completely}. Hence, $E_{\alpha}\left(-\left|\lambda_{j}^{-}\left(A\right)\right|\right)\geq0$
		and
		\begin{equation}
			Tr\left(E_{\alpha}\left(A\right)\right)=\sum_{j=1}^{n}E_{\alpha}\left(\lambda_{j}\left(A\right)\right)\geq0.
		\end{equation}
		We note that $\left|A\right|$ is a nonnegative matrix, and so is any power of $\left|A\right|$, thus $E_{\alpha}\left(\left|A\right|\right)$. Hence, $Tr\left(E_{\alpha}\left(\left|A\right|\right)\right) \ge 0$, and then $0\leq K_{\alpha}\left(G\right)$.
		
		It is clear from Eq.~\eqref{eq:Ka-pn} that $K_\alpha = 1$ if and only if $Tr\left(E^-_{\alpha}\left(A\right)\right) = 0$, if and only if there is no negative closed walks of any length involving any vertices, if and only if there is no negative cycles, i.e., the graph is balanced. While for the lower bound, we require 
		$Tr\left(E^-_{\alpha}\left(A\right)\right) = Tr\left(E^+_{\alpha}\left(A\right)\right)$, 
		which can only happen when the number of negative (unbalanced) closed walks is sufficiently large in a signed graph. 
	\end{proof}
	We now consider again the signed Petersen graphs, specifically the two labelled as c) and d) in Fig.~\ref{Petersen graphs}.
	We recall that although graph d) reaches consensus at a time significantly smaller than graph d), the first has a larger value of $K_{1}$, due to the heavy penalization to walks of relatively large sizes, imposed by the exponential (see section \ref{sec:motivation}).
	We now consider $K_{\alpha}$ as a function of $\alpha$, between these two graphs; see Fig.~\ref{cycles} (right). 
	Specifically, at $\alpha=1$, the graph d) is more balanced than graph c), corresponding to a negative value of the difference between $K_{\alpha}$ of c) minus that of d). This negative difference becomes larger when $\alpha$ drops from $1$, reaching a minimum at about $\alpha=0.84.$ However, after this point, the trend reverses towards positive values,
	reaching the maximum at around $\alpha=0.5.$ 
	At this value of $\alpha$, the penalization of longer cycles is not as heavy, since the larger number of negative hexagons in d) overcome the larger number of negative pentagons in c). 
	If we now correlate the values of $K_{0.5}$ versus $t_{c}$, the squared Pearson correlation coefficient has value $0.981$, which clearly contrasts with the one of $0.924$ for $K_1$, implying that $K_{0.5}$ provides a better indicator of balance in terms of the convergence of the diffusion ($K_{0.5}$ for the five signed Petersen graphs of Fig.~\ref{Petersen graphs}, from a) to e), are: $0.3878$; $0.1973$; $0.1514$; $0.1302$; $0.0787$).

	To gain more insights about the significance of the use of memory to account for balance in signed graphs, let us further explore the changes. When $\alpha$ drops from $1.0$ to about $0.8$, the balance index of the signed Petersen graph d) increases relative to that of graph c). This can be explained by the fact that these graphs are triangle and quadrilateral free, and the smallest cycle is of length five. As we drop initially the value of $\alpha$ from $1$, the contribution of $C_{5}^{-}$ increases, and because graph c) has more of these cycles than d), it is less balanced relative to d). 
	However, as we continue dropping $\alpha$, the contribution of $C_{6}^{-}$ growth significantly. In this case, graph d) overcomes graph c) in the number of $C_{6}^{-}$, which make c) more balanced than
	d) after some critical value and the difference reaches a maximum at around $\alpha=0.5$. 
	Below this value of
	the memory parameter $\alpha$, the longer cycles, namely $C_{8}^{-}$ and $C_{9}^{-}$,
	makes their contribution. 
	In this case, graph c) overcomes d) in the
	number of $C_{8}^{-}$, but d) contains a bit more $C_{9}^{-}$ than c); see Supplementary Material for details. The effect of these longer unbalanced cycles is a further decay of the balance index of both graphs for $\alpha<0.5$. 
	
	\section{Examples of applications}
	
	\subsection{Gene regulatory networks}
	We first consider the gene regulatory networks of \textit{Saccharomyses cerevisiae} (yeast) and of \textit{Bacillus subtilis}, previously studied as signed undirected graphs by Soranzo et al. \cite{soranzo2012decompositions}.
	We maintain the undirected versions of these networks, and consider only their giant connected components. The balance index at $\alpha=1$ indicates that the network of \textit{S. cerevisiae} is slightly more balanced than that of \textit{B. subtilis}, with $K_{1}\approx0.933$ versus $K_{1}\approx0.643$, respectively. 
	We note that the difference between the values of $K_{1}$ for both networks is smaller than $0.3$ and both are not far from $1$, which implies that there is no significant difference in their degree of balance and that they form relatively balanced systems. 
	However, if we allow for an increment of the memory in the system by dropping $\alpha$, then the results change significantly. As can be seen in Fig.~\ref{Altafini data} (left), the difference in balance between the two gene regulatory networks increases up to $0.9$ (of a maximum of $1.0$) when $\alpha$ drops from $1.0$ to $0.62$, where $K_{0.62}\approx0.928$ for \textit{S. cerevisiae} and $K_{0.62}\approx0.014$ for \textit{B. subtilis}. That is, while the gene regulatory network of yeast is highly balanced, the one of \textit{B. subtilis} is extremely unbalanced.
	
	\begin{figure}[h]
		\centering %
		\begin{tabular}{ccc}
			\includegraphics[width=0.3\textwidth]{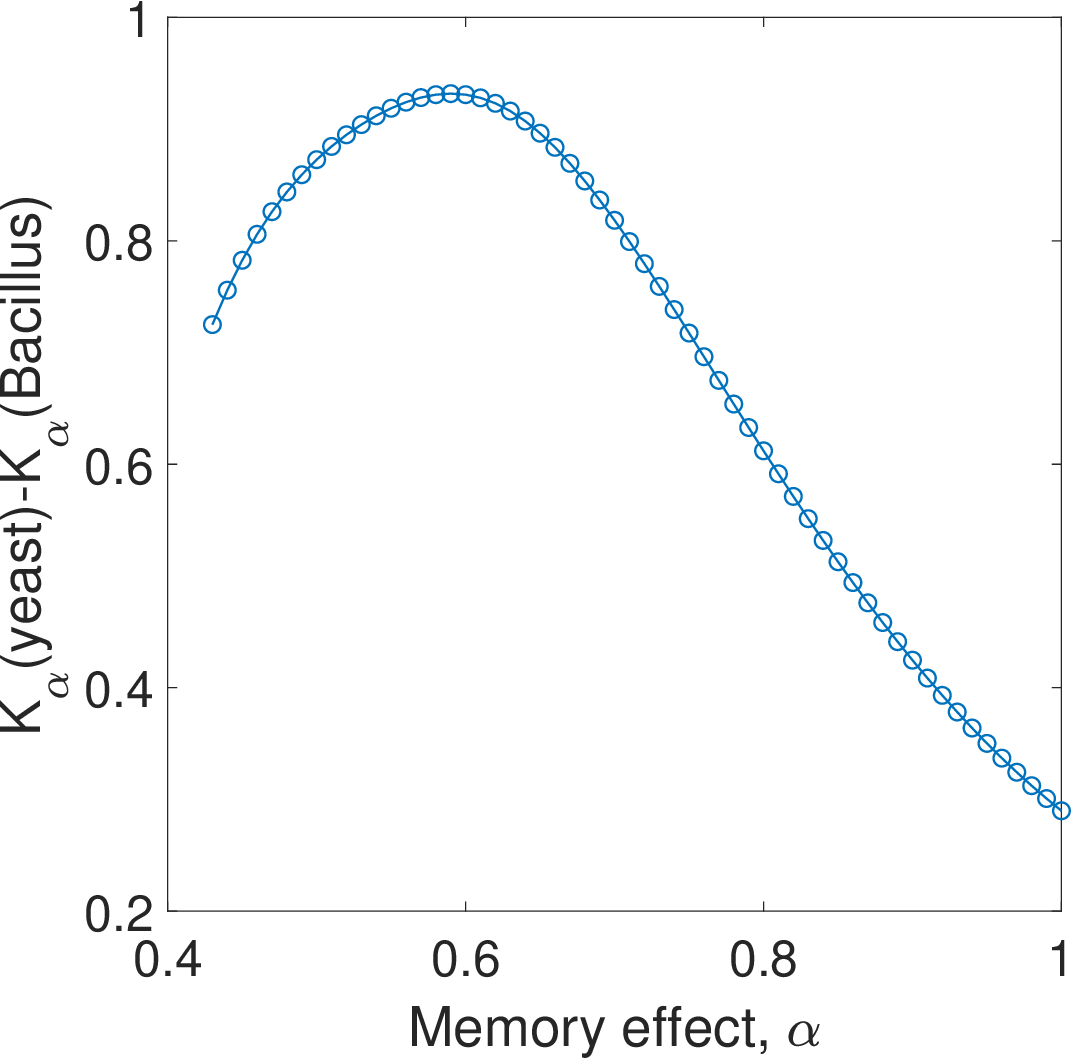}  & \includegraphics[width=0.3\textwidth]{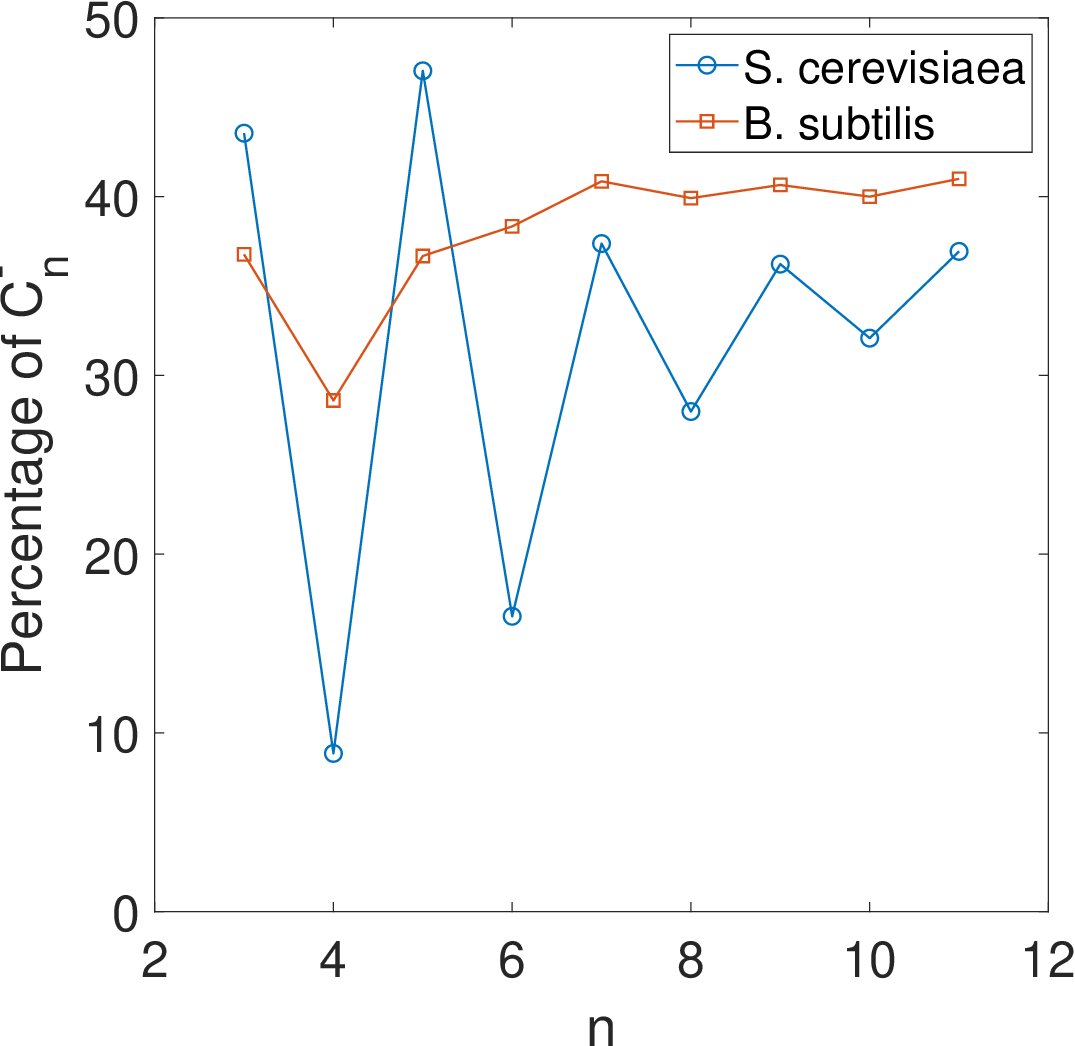}  & \includegraphics[width=0.3\textwidth]{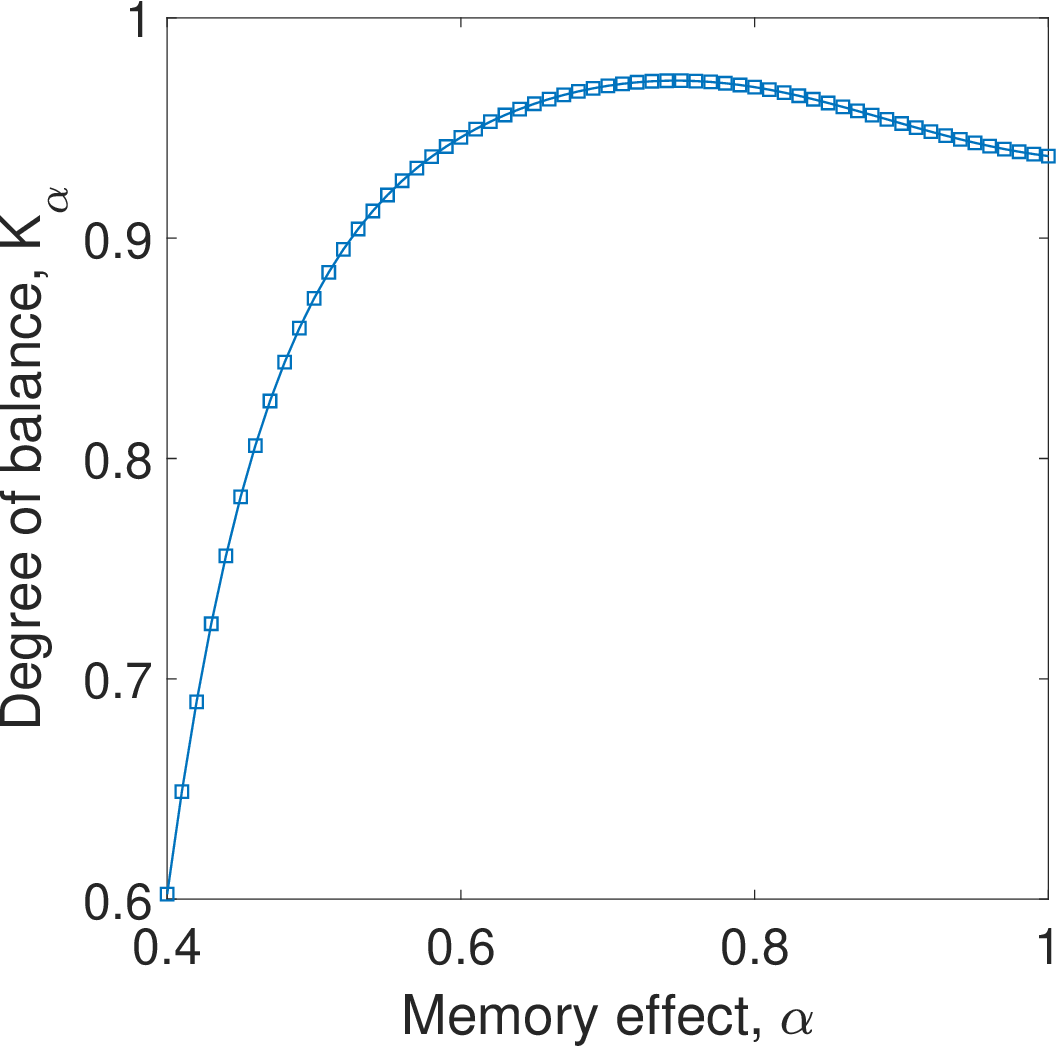} \tabularnewline
		\end{tabular}
		\caption{(left) Plot of the difference in balance $K_{\alpha}$ as a function of $\alpha$ for the gene regulatory networks of yeast and \textit{B. subtilis}. 
			(middle) Plot of the percentages of negative cycles $C_{n}^{-}$ for $3\protect\leq n\protect\leq11$ for the gene regulatory networks of yeast and \textit{B. subtilis}. The percentages are calculated as $100\cdot C_{n}^{-}/\left(C_{n}^{+}+C_{n}^{-}\right).$ 
			(right) Plot of the nonmonotonic change of $K_{\alpha}$ as a function of $\alpha$ for the gene regulatory network of yeast.}
		\label{Altafini data} 
	\end{figure}
	
	This difference is mainly due to the fact that the network of \textit{B. subtilis} has a large number of relatively large unbalanced cycles. 
	We observe that although $C_{n}^{-}$ grows exponentially fast in both networks, it grows faster for the network of \textit{B. subtilis}; see Supplementary Material for details. This can be implied from Fig.~\ref{Altafini data} (middle) where we visualise the percentages of negative cycles of increasing lengths. It can be seen that the network of yeast has relatively more unbalanced triangles and pentagons but significantly less percentage of negative squares than the one of \textit{B. subtilis}.
	This explains why both networks have comparable values of the balance index when $\alpha$ is close to one, i.e., the memory of the system is relatively low although the one of yeast is slightly more balanced than the one of \textit{B. subtilis}. 
	However, when cycles of longer
	lengths ($n\ge 6$) are taken into account, the gene regulatory network of \textit{B. subtilis} has systematically more percentage of unbalanced cycles than the one of yeast. This clearly explains why the network of \textit{B. subtilis} is significantly less balanced than the one of yeast for relatively
	low values of $\alpha$, i.e., the memory of the system increases.
	
	Regarding the change of $K_{\alpha}$ with respect to the drop of the memory parameter $\alpha$ in signed graphs, another interesting characteristic is the possibility of nonmonotonicity; see the case of the gene regulatory network of yeast in Fig.~\ref{Altafini data} (right). 
	Specifically, for the network of yeast, $K_{\alpha}$ increases when $\alpha$ drops from $1.0$ to about $0.75$, and then decays very quickly for values $\alpha<0.75$. 
	In practical terms, this means that there is an ``optimal'' value of the memory that maximizes the degree of balance of this network, and that such value is different from $\alpha=1$.
	The structural explanation for this nonmonotonicity can also be found in the plot in Fig.~\ref{Altafini data} (middle). We observe that there is a significant drop in the percentage of negative squares in this network, which contributes to increasing balance when we drop $\alpha$ from $1.0$ to about $0.75$. However, as value of $\alpha$ decays beyond $0.75$, the longer negative cycles become more important, and the global balance of the network quickly decays.
	
	\subsection{Spatial ecological networks}
	We now study a series of $31$ signed networks representing patterns of spatial (co)occurrence of plants in four major locations in Spain, specifically, Cabo de Gata-Nijar National Park (36.77N, --2.11W), Monegros (41.65N, --0.71W), Sierra de Guara (42.27N, 0.18W) and Ordesa-Monte Perdido National Park (42.63N, --0.11W). The vertices of these networks represent plant species and two vertices form an edge in the graph if the corresponding plants has a spatial association, which was calculated by comparing the number of times that the two species appeared at the same point on the transects. 
	Two plant species share a positive edge if they appeared associated in close region of space, while negative associations correspond to plants appearing separated at a significant distance in space \cite{saiz2017evidence}.
	Therefore, patterns of signed cycles appear in these networks. The meaning of these patterns is self-explained, where a fully positive triangle, for instance, indicates that the three plants have certain type of cooperative relations that allow them to coexist in the same spatial region. A fully negative triangle indicates competitive interactions between the three species that avoids their coexistence in the same location.
	
	We give the average values of $K_{\alpha=1}$ and $K_{\alpha=0.8}$ for the networks in each of the four major locations in Table \ref{data ecological}.
	We also reproduce the values of the mean temperature and precipitation of those regions as reported by Saiz et al.~\cite{saiz2017evidence}.
	The results for $\left\langle K_{\alpha=1}\right\rangle $ are qualitatively similar to those in \cite{saiz2017evidence}, where an index of balance $R$ based on triangles only was used. 
	These results lead to the fact that the balance in those places of higher temperature and lower precipitation is bigger than in those where the temperature is low and the precipitation is high. Both $\left\langle K_{\alpha=1}\right\rangle $ and $R$ identify Monegros as the site with the largest balance and Sierra de Guara as the one more out of balance. 
	However, when we increase the memory of the system by considering a lower value of $\alpha$, e.g., $\left\langle K_{\alpha=0.8}\right\rangle $, a swap on the values of balance of Cabo de Gata and Monegros appears;see Table \ref{data ecological}.
	
	\begin{table}[h]
		\begin{centering}
			\resizebox{\textwidth}{!}{
				\begin{tabular}{|c|c|c|c|c|}
					\hline 
					location  & Temp. ($^{\circ}C$)  & Prec. (mm)  & $\left\langle K_{\alpha=1}\right\rangle $  & $\left\langle K_{\alpha=0.8}\right\rangle $\tabularnewline
					\hline 
					\hline 
					Cabo de Gata  & 24  & 328  & $0.335\pm0.265$  & $0.165\pm0.192$\tabularnewline
					\hline 
					Monegros  & 21  & 360  & $0.376\pm0.221$  & $0.129\pm0.130$\tabularnewline
					\hline 
					Sierra de Guara  & 17  & 927  & $0.0883\pm0.033$  & $0.0061\pm0.0042$\tabularnewline
					\hline 
					Ordesa-Monte Perdido  & 11  & 1485  & $0.0985\pm0.072$  & $0.014\pm0.016$\tabularnewline
					\hline 
				\end{tabular}
			}
			\par\end{centering}
		\caption{Values of the balance degree indices $\left\langle K_{\alpha=1}\right\rangle$ and $\left\langle K_{\alpha=0.8}\right\rangle $ averaged for all the networks in the four main geographic locations studied here. The values of the mean temperature and precipitation in those regions are reported as in Saiz et al. \cite{sadeghi2018some}.}
		\label{data ecological} 
	\end{table}
	
	We visualise the average and standard deviations of $K_{\alpha}$ for $0.4\leq\alpha\leq1$ for the $31$ signed ecological networks grouped in the four major sites under study in Fig.~\ref{plots_ecological} (left). The crossing between the balance rankings of Cabo de Gata and Monegros occurs around $\alpha=0.8$. 
	Furthermore, if we try to explain the degree of balance of these sites by considering a single parameter like the precipitation -- while noticing the risks of doing any correlation for only four points -- we observe some interesting patterns. A power-law fitting of the type: $\left\langle K_{\alpha=1}\right\rangle \sim P^{-2.116}$ gives a correlation coefficient of $r\approx0.88$. Similarly, $R\sim P^{-1.428}$ with $r\approx0.89$. 
	However, when memory effects take place, we obtain: $\left\langle K_{\alpha=0.8}\right\rangle \sim P^{-3.156}$ with $r\approx0.94$. That is, the memory effects increase the amount of variance in the index explained from $77\%$ with $\left\langle K_{\alpha=1}\right\rangle$ to $88\%$ with $\left\langle K_{\alpha=0.8}\right\rangle $.
	
	\begin{figure}[h]
		\centering %
		\begin{tabular}{cc}
			\includegraphics[height=0.23\textheight]{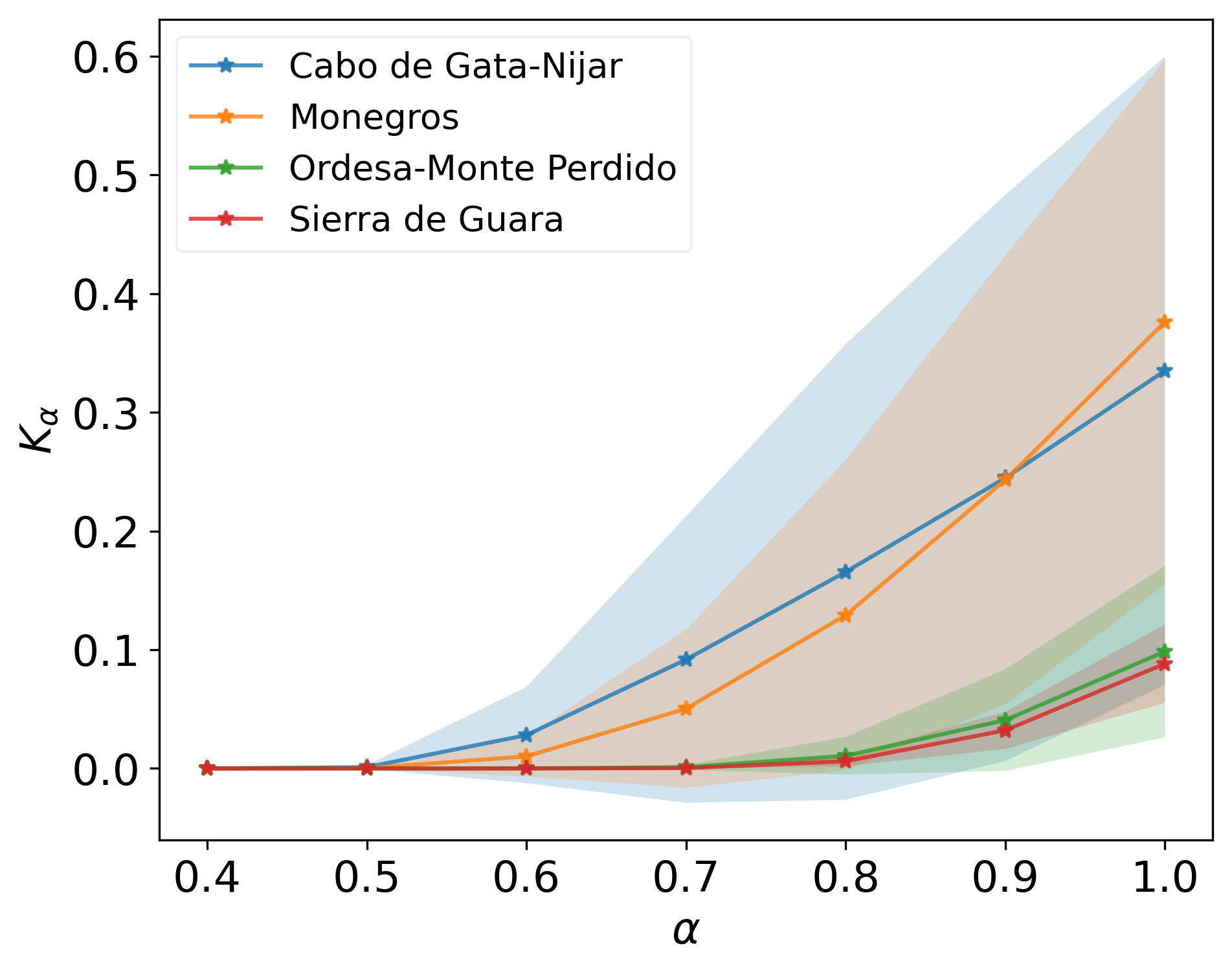}  & \includegraphics[height=0.23\textheight]{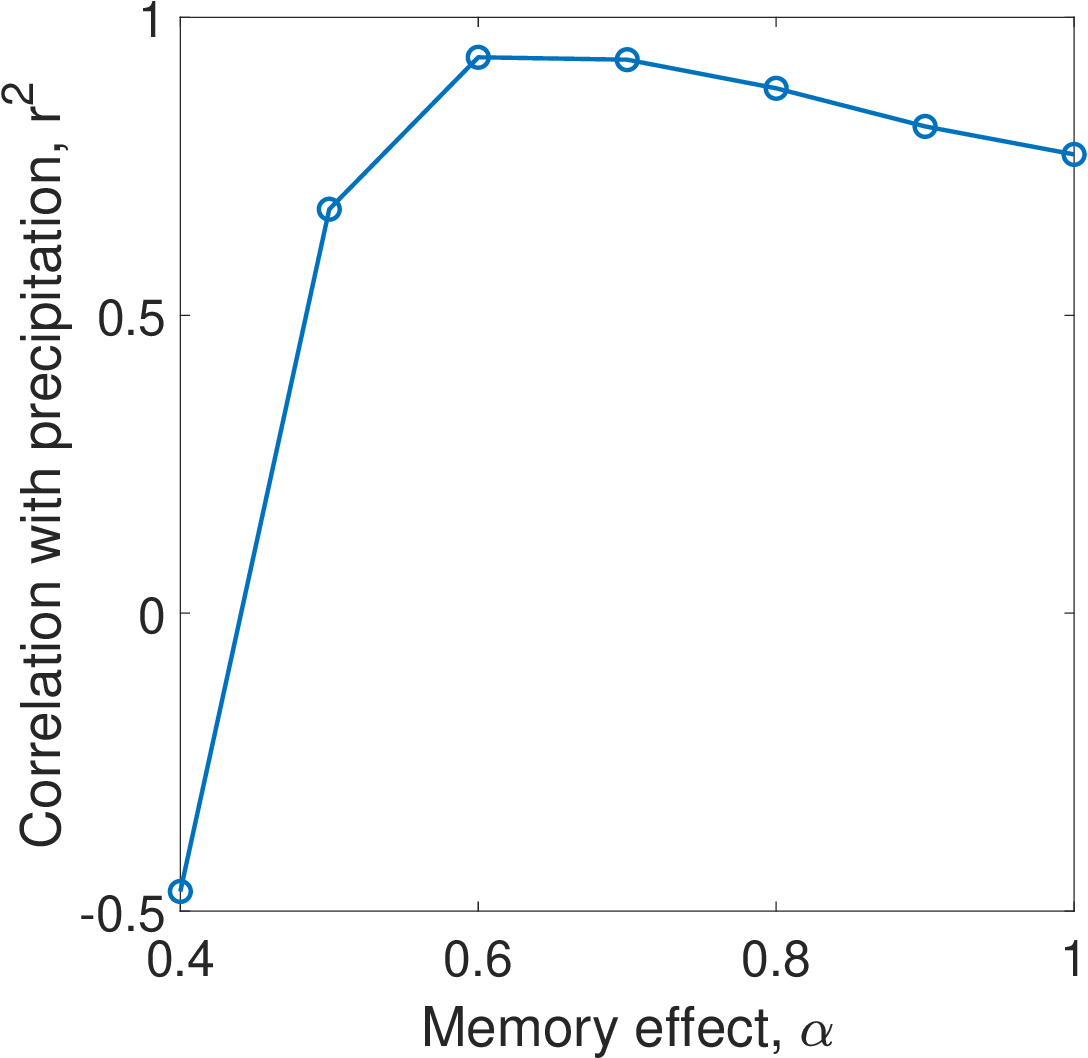} \tabularnewline
		\end{tabular}
		\caption{(left) Change of the average balance $K_{\alpha}$ for $0.4\protect\leq\alpha\protect\leq1$ of the signed ecological networks in the four major Spanish locations studied. (right) Plot of the change of the Pearson correlation coefficient $r^{2}$ of the balance indices with memory $K_{\alpha}$ for $0.4\protect\leq\alpha\protect\leq1$ versus the amount of precipitation in mm on the four major locations studied in this work.}
		\label{plots_ecological} 
	\end{figure}
	
	The question of the existence of an optimal value for the memory effect remains. We obtain the power-law correlation between the balance indices $K_{\alpha}$ for $0.4\leq\alpha\leq1$ and the mean precipitation in the corresponding main locations. The correlation coefficient increases when $\alpha$ drops from $1$ up to $0.6$, and then it decays very quickly; see Fig.~\ref{plots_ecological} (right). This implies that memory effects in ecological systems may have an optimum. However, more research in this area is needed to obtain more conclusive insights about this important question, and we leave it to future work.
	
	\subsection{Social networks in rural villages}
	Finally, we consider a set of social networks constructed from the data of $24696$ people aged $12$ to $93$ years in geographically isolated villages in western Honduras \cite{isakov2019structure}. The vertices of these networks represent residents within each village, and they are connected by a positive (negative) edge if either of them identify the other as a friend (an enemy), while if one identify the other as a friend while the other identify the one as an enemy, we connect them by a negative edge. 
	We note that the case that they have no opinion of each other is also allowed. By design, the networks are solely within-village networks, and we select $11$ of them for our analysis, labelled as $A$ up to $K$. Cycles of various lengths can frequently occur in such social networks, positive or negative. Corresponding to balance theory \cite{cartwright1956structural}, positive cycles indicate that the residents can be partitioned into one or two communities without conflicts inside each community, while negative cycles indicate the existence of conflicts of the relationships between residents.
	
	\begin{figure}[htbp]
		\centering \includegraphics[height=0.26\textheight]{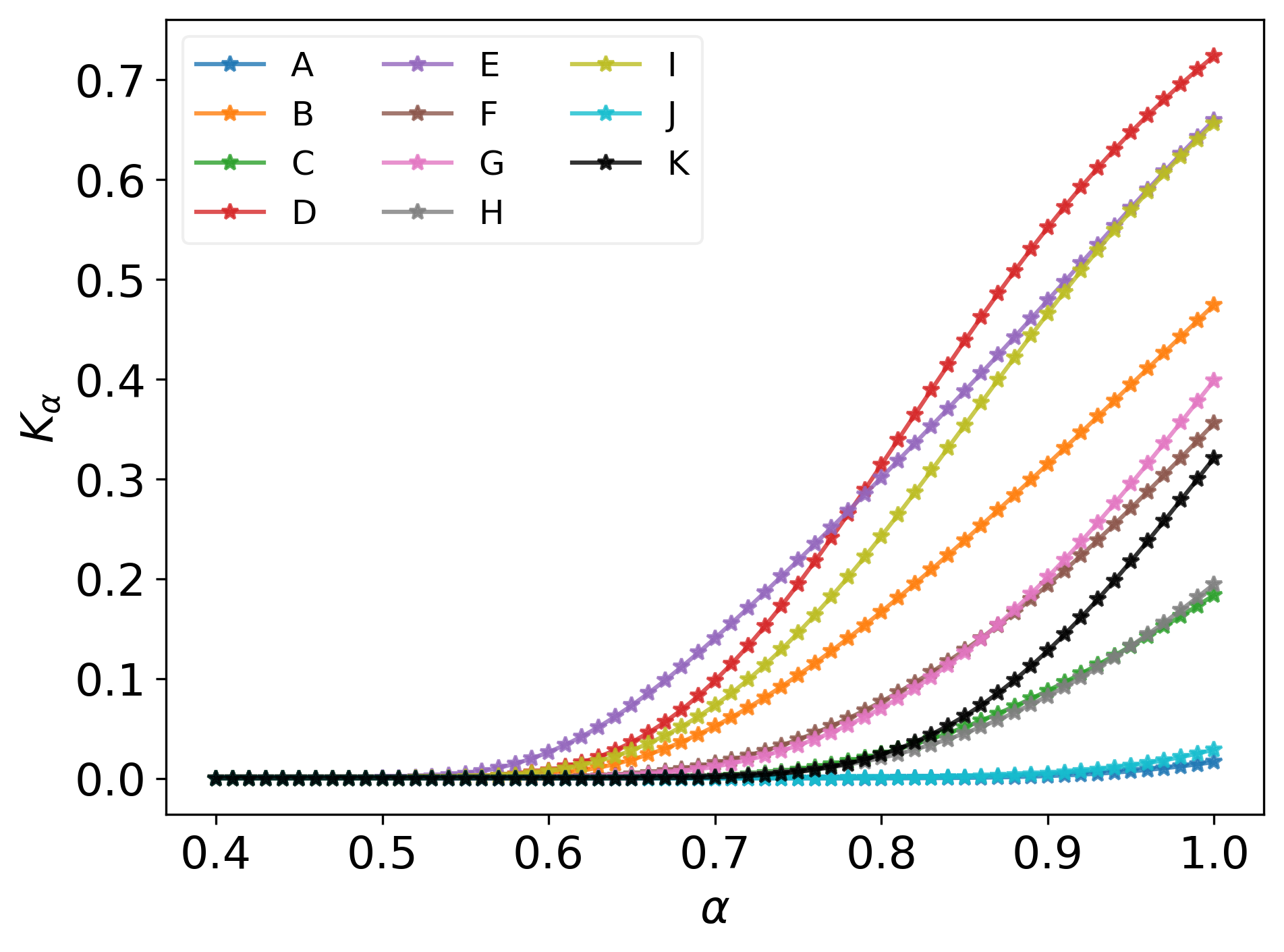} \caption{Plot of the balance index $K_{\alpha}$ as a function for $\alpha$ for the $11$ social networks of rural villages.}
		\label{fig:ruralV-Ka} 
	\end{figure}
	
	We consider the balance index $K_{\alpha}$ for $0.4\le\alpha\le1$ for the $11$ signed social networks; see Fig.~\ref{fig:ruralV-Ka}.
	We find that all networks are not completely balanced: all $11$ networks reaches the maximum value of the balance index $K_{\alpha}$ at $\alpha=1$, and $K_{\alpha}$ quickly decreases as $\alpha$ deviates from $1$, where the values are almost $0$ for all networks at $\alpha=0.5$.
	For example, village D has the maximum index value in all $11$ networks at $\alpha=1$, with $K_{1}\approx0.723$, which implies that the network is close to being balanced. However, it becomes less than $0.5$ at $\alpha=0.8$, and continue decreasing as we increase the memory effect through parameter $\alpha$. These imply the abundance of long negative cycles in these social networks, which is consistent with the results in \cite{isakov2019structure}, such as the homophily of negative relationships.
	
	\begin{figure}[h]
		\centering %
		\begin{tabular}{cc}
			\includegraphics[height=0.23\textheight]{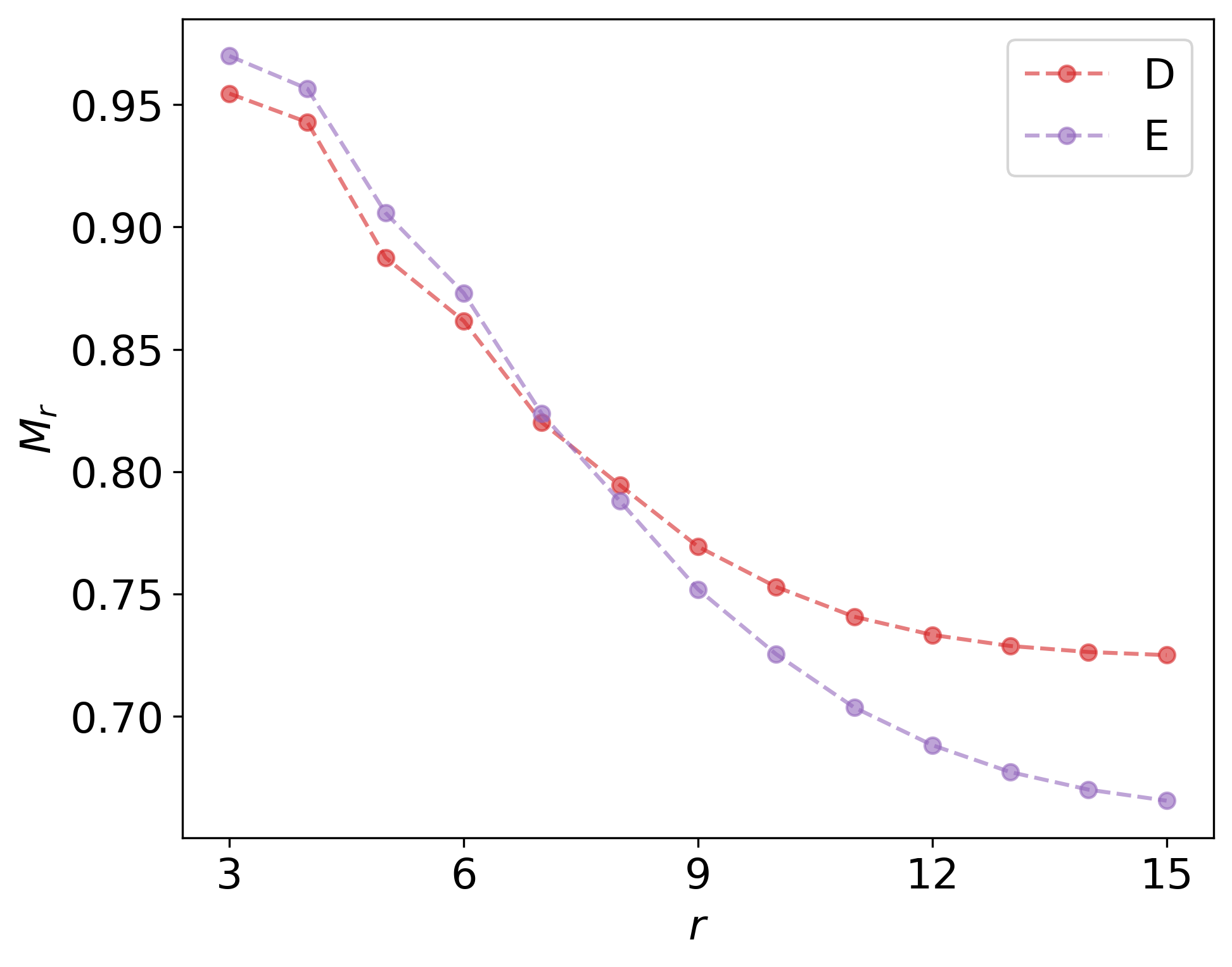}  & \includegraphics[height=0.23\textheight]{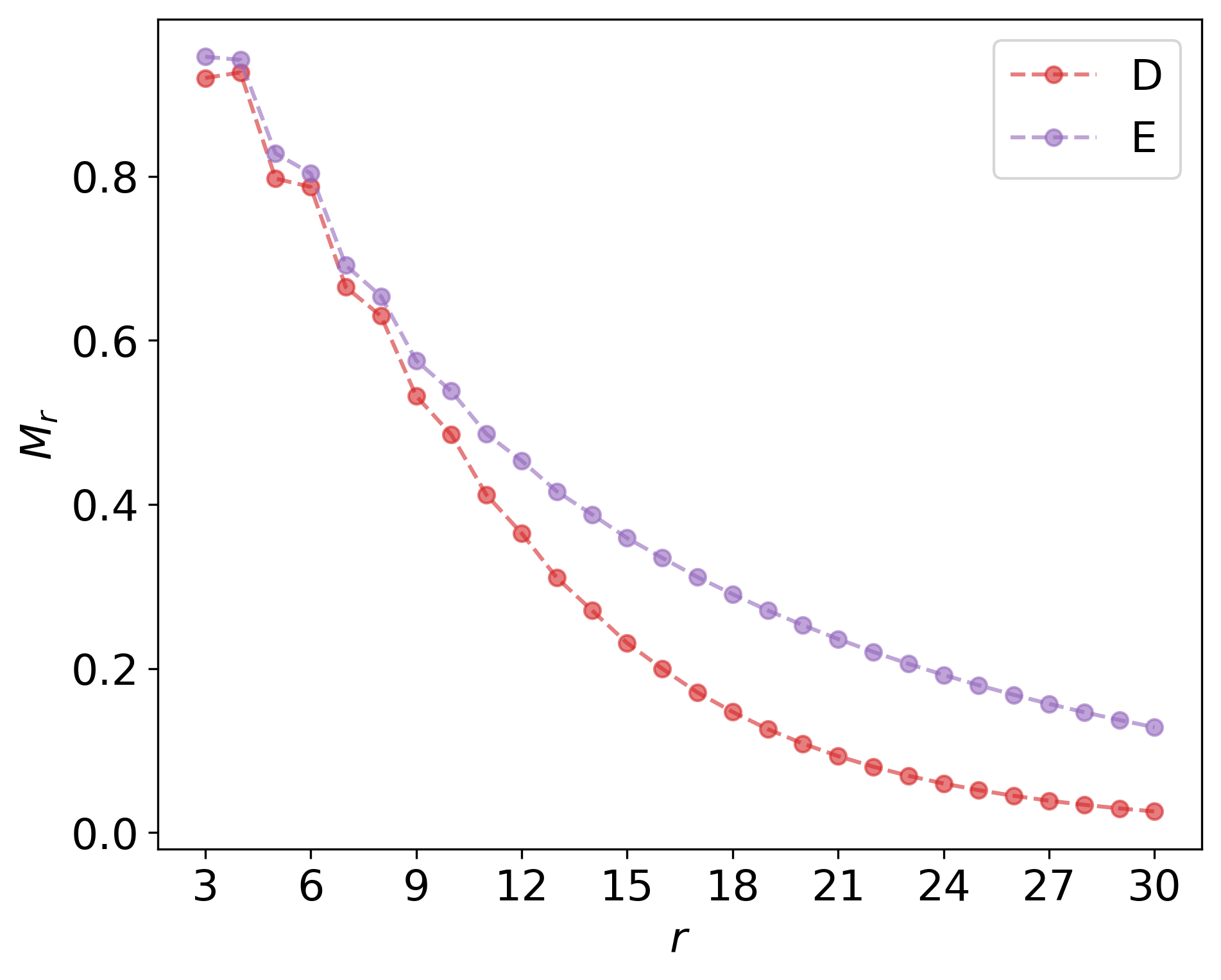} \tabularnewline
		\end{tabular}\caption{Plot of the truncated index $M_{r}$ as a function for $r$ for the
			social networks of villages $D$ and $E$ when $\alpha=1$ (left)
			and $\alpha=0.6$ (right)}
		\label{fig:ruralV-moments} 
	\end{figure}
	
	Specifically, we observe a clear crossing between the change of the index values of village D and that of village E as $\alpha$ deviates from $1$. In order to understand the differences between the balance
	of the villages $D$ and $E$, we start by defining the following truncated series, $M_{r}^{s}\coloneqq\sum_{k=0}^{r}Tr\left(A^{k}\right)/\varGamma\left(\alpha k+1\right)$
	which accounts for the signed contributions of the different spectral
	moments; $M_{r}^{u}\coloneqq\sum_{k=0}^{r}Tr\left(\left|A\right|^{k}\right)/\varGamma\left(\alpha k+1\right)$ 
	which accounts for the total contribution of closed walks, and $M_{r}=M_{r}^{s}/M_{r}^{u}$.
	Obviously, $M_{r\rightarrow\infty}=M_{r\rightarrow\infty}^{s}/M_{r\rightarrow\infty}^{t}=Tr\left(E_{\alpha}\left(A\right)\right)/Tr\left(E_{\alpha}\left(\left|A\right|\right)\right)$ recovers the balance index $K_{\alpha}$. 
	We start by truncating the series at $r=3$ which is where the first signed cycles appear, and then continue increasing $r$. First, we plot the results for the two graphs when $\alpha=1$ in the left of Fig.~\ref{fig:ruralV-moments}. Hence, the network of village $E$ (purple circles) appears to be more balanced than the one of village $D$ if we truncate the sum of spectral moments below $r=8$.
	Indeed, a simple index based only on triangles indicates that $E$ is more balanced than $D$. At about $r=8$ the cumulative sum of moments for graph $D$ become larger than that of graph $E$, indicating that now the former graph is more balanced. 
	The reason for this swap in the balance order is not directly caused by the larger number of longer signed cycles in one over the other, as we have seen in previous examples, but due to the fact that for graph $D$ the ratio of the cumulative sum of moments of length smaller than $8$ is smaller than that for the graph $E$. However, increasing this sum to higher-order moments makes it bigger for graph $D$ than to graph $E$. For example, the ratio $M_{7}\left(D\right)\approx2336/2850\approx0.8197$ 
	which is smaller than $M_{7}\left(E\right)\approx4442/5394\approx0.8236$. 
	However, $M_{8}\left(D\right)\approx(2336+239)/(2850+393)\approx0.7941$ 
	is smaller than $M_{8}\left(E\right)\approx(4442+815)/(5394+1278)\approx0.7880$, 
	which is independent of the fact that $(239/393)<(815/1278)$,
	but depending on the rates on which the numerator and denominator of the $M_{7}\left(D\right)$ and $M_{7}\left(E\right)$ growth by the addition of the individual terms.
	
	This effect previously seen for $\alpha=1$ disappears when we increases the memory effect; see the right of Fig.~\ref{fig:ruralV-moments}. Specifically, penalizing less the walks increases the difference in balance in favor of graph $E$ relative to $D$. Therefore, we may consider the fact that the factorial penalization points out to graph $D$ as more balance than $E$ as an artifact of this type of penalization, which indeed is solved when the memory of the system increases. This emphasize the importance of selecting an optimal memory effect parameter $\alpha$
	to understand the level of balance of the signed social networks.
	
	\section{A useful approximation}
	We know that
	\begin{equation}
		Tr\left(E_{\alpha}\left(\gamma A\right)\right)=\sum_{j=1}^{n}E_{\alpha}\left(\gamma\lambda_{j}\right).
	\end{equation}
	Let $\lambda_{1}^{\left(m_{1}\right)}>\lambda_{2}^{\left(m_{2}\right)}>\ldots>\lambda_{r}^{\left(m_{r}\right)}$
	be the distinct eigenvalues of $A$ together with their multiplicities $m_{i}$. For $\alpha$ relatively low, the function $E_{\alpha}\left(z\right)$ grows extremely fast with the values $z$. Therefore, for relatively small values of $\alpha$ the difference $\lambda_{1}>\lambda_{2}$ is magnified by $E_{\alpha}\left(\lambda_{1}\right)\ggg E_{\alpha}\left(\lambda_{2}\right),$ which implies that 
	\begin{equation}
		\underset{\alpha\rightarrow0}{\lim}Tr\left(E_{\alpha}\left(\gamma A\right)\right)=m_{1}E_{\alpha}\left(\gamma\lambda_{1}\right).
	\end{equation}
	If the eigenvalues of $\left|A\right|$ are $\mu_{1}>\mu_{2}\geq\ldots\geq\mu_{n}$
	we will have that
	\begin{equation}
		\tilde{K}_{\alpha}\coloneqq\underset{\alpha\rightarrow0}{\lim}K_{\alpha}\left(G_{s}\right)=\dfrac{m_{1}E_{\alpha}\left(\gamma\lambda_{1}\right)}{E_{\alpha}\left(\gamma\mu_{1}\right)}.
	\end{equation}
	
	We experimentally verify the goodness of this approximation in some of the networks we have studied previously, specifically the signed networks of the rural villages. In Fig.~\ref{approximation} (left), we plot the relative error in the balance index with memory $K_{\alpha}$ when approximated by $\tilde{K}_{\alpha}$ for the $11$ signed networks of the rural villages. We observe that for values of $\alpha$ close to $1$, the relative error is relatively large for most of the networks, with values of up to $0.7$ for village K. However, the error significantly drops when $\alpha$ changes systematically to $0$, and in particular, for $\alpha=0.5$ the average relative error for the $11$ networks is $0.071$ with only network D having a relatively large error of about $0.3$. When $\alpha=0.4$ all networks display error below $0.1$, where most of them have extremely low errors, e.g., below $10^{-10}.$ 
	\begin{figure}[ht]
		\centering
		\begin{tabular}{cc}
			\includegraphics[height=.23\textheight]{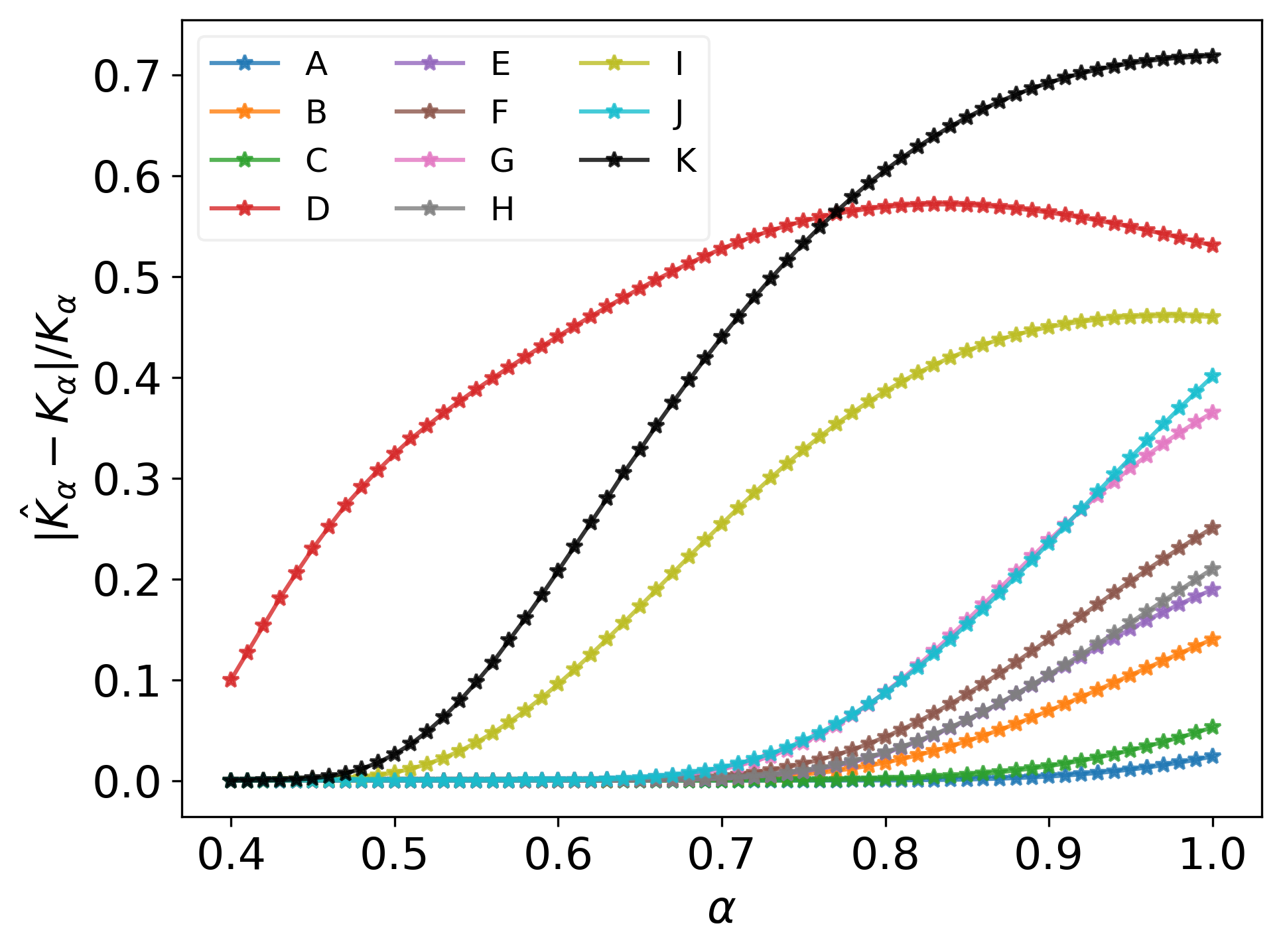} & \includegraphics[height=.23\textheight]{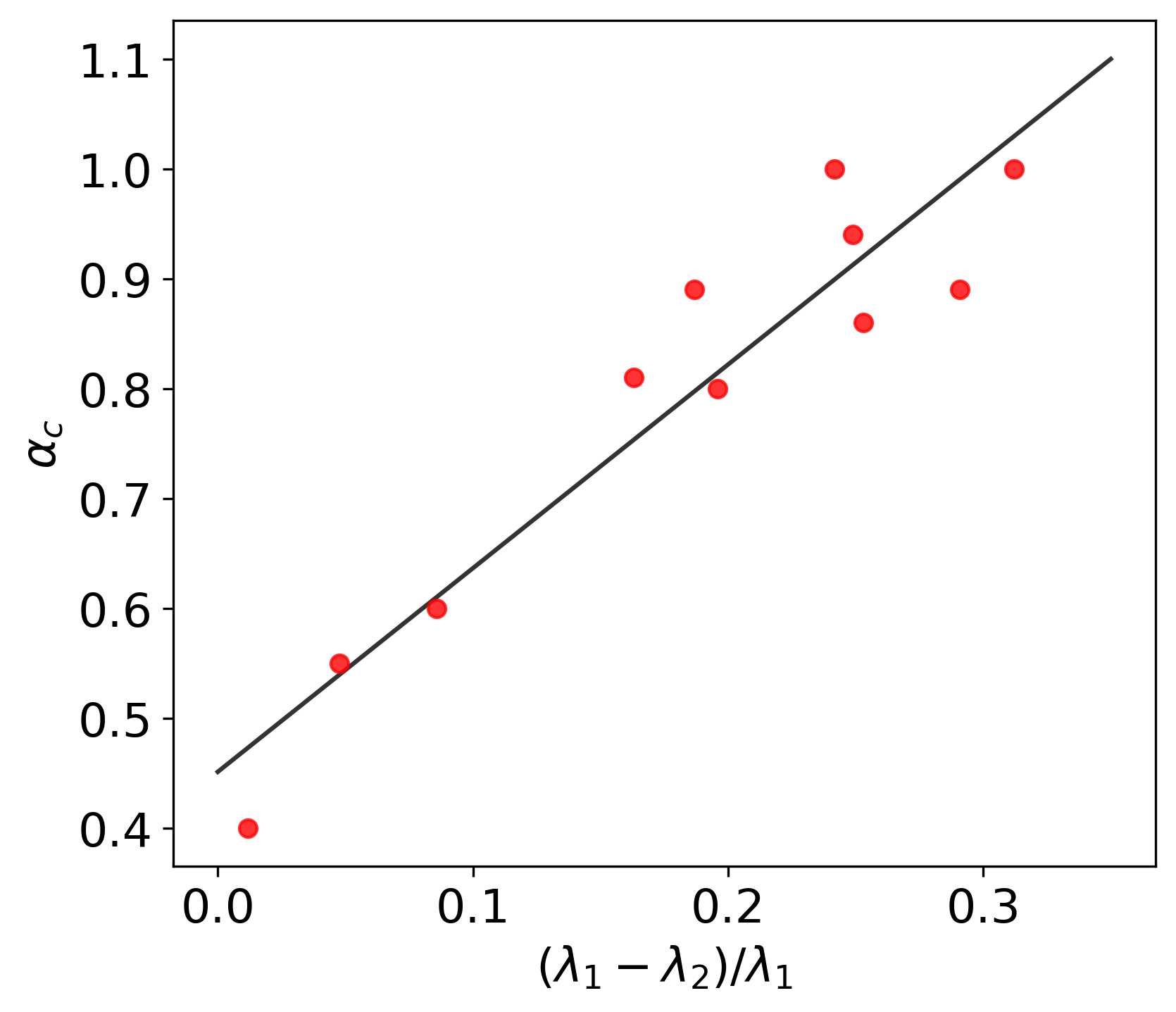}
		\end{tabular}
		\caption{(left) Plot of the relative error in the approximation of $K_{\alpha}$ by $\tilde{K}_{\alpha}$. (right) Plot of the values of $\alpha$ for which the relative error in the approximation drops below $0.1$, $\alpha_c$, versus the relative spectral gap, together with the fitted regression line.}
		\label{approximation}
	\end{figure}
	
	As in the derivation, the main driver for this approximation is the spectral gap of the adjacency matrix of the signed graph, i.e., $\lambda_{1}-\lambda_{2}$. Here we obtain the value of $\alpha$ for which the relative error in the approximation of $K_{\alpha}$ by $\tilde{K}_{\alpha}$ drops below $0.1$ (for illustrative purposes), denoted by $\alpha_{c}$. In Fig.~\ref{approximation} (right), we plot the values of $\alpha_{c}$ for every network as a function of the relative spectral gap $\left(\lambda_{1}-\lambda_{2}\right)/\lambda_{1}$.
	We observe that increasing the spectral gap makes the approximation works better even for relatively large values of $\alpha$ (Pearson correlation: $0.942$). In contrast, for those networks like the one of village D, where $\lambda_{1}\approx6.333$ and $\lambda_{2}\approx6.258$, it is hard to converge even for relatively low values of $\alpha$. However, the trend holds where the relative error of this approximation is significantly lower for values of $\alpha$ relatively low than that for the value of $\alpha=1$, where the balance index corresponds to the exponential.
	


	\section*{Acknowledgments}
	We would like to acknowledge Dr. H. Saiz and Prof. C. Altafini for sharing datasets used in this work. E.E.~acknowledges support from the Maria de Maeztu project CEX2021-001164-M funded by the MCIN/AEI/10.13039/501100011033. Y.T.~is funded by the Wallenberg Initiative on Networks and Quantum Information (WINQ).

	\bibliographystyle{plain}
	\bibliography{refs}

\begin{thebibliography}{10}

\bibitem{alfakih2018euclidean}
A.~Alfakih.
\newblock {\em Euclidean distance matrices and their applications in rigidity
  theory}.
\newblock Springer, 2018.

\bibitem{Altafini_2012_opinion}
C.~Altafini.
\newblock Dynamics of opinion forming in structurally balanced social networks.
\newblock {\em PLoS ONE}, 7(6):1--9, 2012.

\bibitem{altafini2013biconsensus}
C.~Altafini.
\newblock Consensus problems on networks with antagonistic interactions.
\newblock {\em IEEE Transactions on Automatic Control}, 58(4):935--946, 2013.

\bibitem{arrigo2021mittag}
Francesca Arrigo and Fabio Durastante.
\newblock Mittag--leffler functions and their applications in network science.
\newblock {\em SIAM Journal on Matrix Analysis and Applications},
  42(4):1581--1601, 2021.

\bibitem{Atay_signedCheeger_2020}
F.~Atay and S.~Liu.
\newblock Cheeger constants, structural balance, and spectral clustering
  analysis for signed graphs.
\newblock {\em Discrete Math.}, 343(1):111616, 2020.

\bibitem{balaji2007euclidean}
R.~Balaji and R.~Bapat.
\newblock On euclidean distance matrices.
\newblock {\em Linear Algebra Appl.}, 424(1):108--117, 2007.

\bibitem{benzi2020matrix}
Michele Benzi and Paola Boito.
\newblock Matrix functions in network analysis.
\newblock {\em GAMM-Mitteilungen}, 43(3):e202000012, 2020.

\bibitem{berge2001theory}
C.~Berge.
\newblock {\em The theory of graphs}.
\newblock Courier Corporation, 2001.

\bibitem{cartwright1956structural}
Dorwin Cartwright and Frank Harary.
\newblock Structural balance: a generalization of heider's theory.
\newblock {\em Psychological review}, 63(5):277, 1956.

\bibitem{diaz2024signed}
F.~Diaz-Diaz and E.~Estrada.
\newblock Signed graphs in data sciences via communicability geometry.
\newblock {\em arXiv preprint arXiv:2403.07493}, 2024.

\bibitem{estrada2024conservative}
E.~Estrada.
\newblock Conservative vs. non-conservative diusion towards a target in a
  networked environment.
\newblock In {\em The Target Problem}. Springer, Berlin, 2024.

\bibitem{estrada2012communicability}
Ernesto Estrada.
\newblock The communicability distance in graphs.
\newblock {\em Linear Algebra and its Applications}, 436(11):4317--4328, 2012.

\bibitem{estrada2019rethinking}
Ernesto Estrada.
\newblock Rethinking structural balance in signed social networks.
\newblock {\em Discrete Applied Mathematics}, 268:70--90, 2019.

\bibitem{estrada2022many}
Ernesto Estrada.
\newblock The many facets of the estrada indices of graphs and networks.
\newblock {\em SeMA Journal}, 79(1):57--125, 2022.

\bibitem{estrada2024communicability}
Ernesto Estrada.
\newblock Communicability cosine distance: similarity and symmetry in
  graphs/networks.
\newblock {\em Computational and Applied Mathematics}, 43(1):49, 2024.

\bibitem{estrada2014walk}
Ernesto Estrada and Michele Benzi.
\newblock Walk-based measure of balance in signed networks: Detecting lack of
  balance in social networks.
\newblock {\em Physical Review E}, 90(4):042802, 2014.

\bibitem{estrada2023network}
Ernesto Estrada, Jes{\'u}s G{\'o}mez-Garde{\~n}es, and Lucas Lacasa.
\newblock Network bypasses sustain complexity.
\newblock {\em Proceedings of the National Academy of Sciences},
  120(31):e2305001120, 2023.

\bibitem{estrada2008communicability}
Ernesto Estrada and Naomichi Hatano.
\newblock Communicability in complex networks.
\newblock {\em Physical Review E}, 77(3):036111, 2008.

\bibitem{estrada2016communicability}
Ernesto Estrada and Naomichi Hatano.
\newblock Communicability angle and the spatial efficiency of networks.
\newblock {\em SIAM Review}, 58(4):692--715, 2016.

\bibitem{estrada2012physics}
Ernesto Estrada, Naomichi Hatano, and Michele Benzi.
\newblock The physics of communicability in complex networks.
\newblock {\em Physics reports}, 514(3):89--119, 2012.

\bibitem{estrada2010network}
Ernesto Estrada and Desmond~J Higham.
\newblock Network properties revealed through matrix functions.
\newblock {\em SIAM review}, 52(4):696--714, 2010.

\bibitem{estrada2005subgraph}
Ernesto Estrada and Juan~A Rodriguez-Velazquez.
\newblock Subgraph centrality in complex networks.
\newblock {\em Physical Review E}, 71(5):056103, 2005.

\bibitem{estrada2014hyperspherical}
Ernesto Estrada, MG~Sanchez-Lirola, and Jos{\'e}~Antonio De~La~Pe{\~n}a.
\newblock Hyperspherical embedding of graphs and networks in communicability
  spaces.
\newblock {\em Discrete Applied Mathematics}, 176:53--77, 2014.

\bibitem{Facchetti_2011_large}
G.~Facchetti, G.~Iacono, and C.~Altafini.
\newblock Computing global structural balance in large-scale signed social
  networks.
\newblock {\em Proc. Natl. Acad. Sci.}, 108(52):20953--20958, 2011.

\bibitem{festinger1949analysis}
L.~Festinger.
\newblock The analysis of sociograms using matrix algebra.
\newblock {\em Hum. Relat.}, 2(2):153--158, 1949.

\bibitem{ghosh2024non}
Rumi Ghosh, Kristina Lerman, Tawan Surachawala, Konstatin Voevodski, and
  Shanghua Teng.
\newblock Non-conservative diffusion and its application to social network
  analysis.
\newblock {\em Journal of Complex Networks}, 12(1):cnae006, 2024.

\bibitem{giscard2017index}
P.~Giscard, P.~Rochet, and R.~Wilson.
\newblock Evaluating balance on social networks from their simple cycles.
\newblock {\em J. Complex Netw.}, 5(5):750--775, 05 2017.

\bibitem{gower1985properties}
J.~Gower.
\newblock Properties of euclidean and non-euclidean distance matrices.
\newblock {\em Linear Algebra Appl.}, 67:81--97, 1985.

\bibitem{harary_1953_balance}
F.~Harary.
\newblock On the notion of balance of a signed graph.
\newblock {\em Michigan Math. J.}, 2(2):143--146, 1953.

\bibitem{heider_1946_psychology}
F.~Heider.
\newblock Attitudes and cognitive organization.
\newblock {\em J. Psychol.}, 21(1):107--112, 1946.

\bibitem{higham2008functions}
Nicholas~J Higham.
\newblock {\em Functions of matrices: theory and computation}.
\newblock SIAM, 2008.

\bibitem{isakov2019structure}
Alexander Isakov, James~H Fowler, Edoardo~M Airoldi, and Nicholas~A Christakis.
\newblock The structure of negative social ties in rural village networks.
\newblock {\em Sociological science}, 6:197, 2019.

\bibitem{katz1953new}
Leo Katz.
\newblock A new status index derived from sociometric analysis.
\newblock {\em Psychometrika}, 18(1):39--43, 1953.

\bibitem{kirkly2019index}
A.~Kirkley, G.~Cantwell, and M.~Newman.
\newblock Balance in signed networks.
\newblock {\em Phys. Rev. E}, 99:012320, Jan 2019.

\bibitem{krislock2012euclidean}
N.~Krislock and H.~Wolkowicz.
\newblock {\em Euclidean distance matrices and applications}.
\newblock Springer, 2012.

\bibitem{lee2019transient}
Chul-Ho Lee, Srinivas Tenneti, and Do~Young Eun.
\newblock Transient dynamics of epidemic spreading and its mitigation on large
  networks.
\newblock In {\em Proceedings of the twentieth ACM international symposium on
  mobile ad hoc networking and computing}, pages 191--200, 2019.

\bibitem{lerman2012network}
Kristina Lerman and Rumi Ghosh.
\newblock Network structure, topology, and dynamics in generalized models of
  synchronization.
\newblock {\em Physical Review E}, 86(2):026108, 2012.

\bibitem{martin2023fractional}
Andr{\'e}s Mart{\'\i}n and Ernesto Estrada.
\newblock Fractional-modified bessel function of the first kind of integer
  order.
\newblock {\em Mathematics}, 11(7):1630, 2023.

\bibitem{mithai2012cycle}
A.~Mathai and T.~Zalavsky.
\newblock On adjacency matrices and descriptors of signed cycle graphs.
\newblock {\em J. Comb. Inf. Syst. Sci.}, 37(2-4):369--382, 2012.

\bibitem{odibat2006approximations}
Zaid Odibat.
\newblock Approximations of fractional integrals and caputo fractional
  derivatives.
\newblock {\em Applied Mathematics and Computation}, 178(2):527--533, 2006.

\bibitem{polard1948completely}
H~Polard.
\newblock The completely monotonic character of the mittag-leffler function.
\newblock {\em Bull. Am. Math. Soc}, 52:908--910, 1948.

\bibitem{sadeghi2018some}
Amir Sadeghi and Jo{\~a}o~R Cardoso.
\newblock Some notes on properties of the matrix mittag-leffler function.
\newblock {\em Applied Mathematics and Computation}, 338:733--738, 2018.

\bibitem{saiz2017evidence}
Hugo Saiz, Jes{\'u}s G{\'o}mez-Garde{\~n}es, Paloma Nuche, Andrea Gir{\'o}n,
  Yolanda Pueyo, and Concepci{\'o}n~L Alados.
\newblock Evidence of structural balance in spatial ecological networks.
\newblock {\em Ecography}, 40(6):733--741, 2017.

\bibitem{seguin2018navigation}
Caio Seguin, Martijn~P Van Den~Heuvel, and Andrew Zalesky.
\newblock Navigation of brain networks.
\newblock {\em Proceedings of the National Academy of Sciences},
  115(24):6297--6302, 2018.

\bibitem{singh2017index}
R.~Singh and B.~Adhikari.
\newblock Measuring the balance of signed networks and its application to sign
  prediction.
\newblock {\em J. Stat. Mech. Theory Exp.}, 2017(6):063302, jun 2017.

\bibitem{soranzo2012decompositions}
Nicola Soranzo, Fahimeh Ramezani, Giovanni Iacono, and Claudio Altafini.
\newblock Decompositions of large-scale biological systems based on dynamical
  properties.
\newblock {\em Bioinformatics}, 28(1):76--83, 2012.

\bibitem{tarazaga1996circum}
P~Tarazaga, T.~Hayden, and J.~Wells.
\newblock Circum-euclidean distance matrices and faces.
\newblock {\em Linear Algebra Appl.}, 232:77--96, 1996.

\bibitem{tian_2022_thesis}
Y.~Tian.
\newblock {\em Role Extraction, Dynamics, and Optimisation on Networks}.
\newblock {PhD} thesis, University of Oxford, October 2022.
\newblock Available at
  \url{https://ora.ox.ac.uk/objects/uuid:8145297d-3f88-4d67-9c34-575beb1a4c6c}.

\bibitem{tian2024sign}
Y.~Tian and R.~Lambiotte.
\newblock Spreading and structural balance on signed networks.
\newblock {\em SIAM J. Appl. Dyn. Syst.}, 23(1):50--80, 2024.

\bibitem{tian_2021_role}
Y.~Tian, S.~Lautz, A.~Wallis, and R.~Lambiotte.
\newblock Extracting complements and substitutes from sales data: a network
  perspective.
\newblock {\em EPJ Data Sci.}, 10(1):45, 2021.

\bibitem{zaslavsky1982signed}
Thomas Zaslavsky.
\newblock Signed graphs.
\newblock {\em Discrete Applied Mathematics}, 4(1):47--74, 1982.

\bibitem{zaslavsky2012six}
Thomas Zaslavsky.
\newblock Six signed petersen graphs, and their automorphisms.
\newblock {\em Discrete Mathematics}, 312(9):1558--1583, 2012.

\end{thebibliography}
	
	\newpage
	\appendix
	\section{Tables}
	We give more details of the results from the signed Petersen graphs in Tables \ref{Results_Petersen} and \ref{Cycles Petersen}.  
	\begin{table}[htbp]
		\begin{centering}
			\begin{tabular}{|c|c|c|c|c|}
				\hline 
				graph  & $t_{c}$  & $C_{5}^{-}$  & $C_{6}^{-}$  & $K\left(G\right)$\tabularnewline
				\hline 
				\hline 
				a  & 48  & 4  & 4  & 0.968\tabularnewline
				\hline 
				b  & 24  & 6  & 6  & 0.951\tabularnewline
				\hline 
				c  & 22  & 8  & 4  & 0.941\tabularnewline
				\hline 
				d  & 14  & 6  & 10  & 0.947\tabularnewline
				\hline 
				e  & 11  & 12  & 0  & 0.919\tabularnewline
				\hline 
			\end{tabular}
			\par\end{centering}
		\caption{Values of the time of consensus $t_{c}$ using Altafini's model \cite{altafini2013biconsensus}
			for the graphs in Fig. \ref{Petersen graphs} as well as the values
			of the number of negative cycles of length 5 $C_{5}^{-}$ and 6 $C_{6}^{-}$
			in those graphs. The values of the balance degree index $K\left(G\right)$
			are also given for each graph. }
		
		\label{Results_Petersen} 
	\end{table}
		
	
	\begin{table}[htbp]
		\begin{centering}
			\begin{tabular}{|c|c|c|c|c|}
				\hline 
				& \multicolumn{2}{c|}{\textit{graph c}} & \multicolumn{2}{c|}{\textit{graph d}}\tabularnewline
				\hline 
				$n$  & $C_{n}^{+}$  & $C_{n}^{-}$  & $C_{n}^{+}$  & $C_{n}^{-}$\tabularnewline
				\hline 
				\hline 
				5  & 4  & 8  & 6  & 6\tabularnewline
				\hline 
				6  & 6  & 4  & 0  & 10\tabularnewline
				\hline 
				8  & 7  & 8  & 15  & 0\tabularnewline
				\hline 
				9  & 12  & 8  & 10  & 10\tabularnewline
				\hline 
			\end{tabular}
			\par\end{centering}
		\caption{Values of the number of positive $C_{n}^{+}$ and negative $C_{n}^{-}$ cycles of length $5\protect\leq n\protect\leq8$ for the signed Petersen graphs c) and d) of Fig.~\ref{Petersen graphs}.}
		
		\label{Cycles Petersen} 
	\end{table}

	We summarise the exact numbers of positive and negative cycles in the gene regulatory networks in Table \ref{Cycles Altafini}.  
	\begin{table}[htbp]
		\begin{centering}
			\begin{tabular}{|c|c|c|c|c|}
				\hline 
				& \multicolumn{2}{c|}{\textit{S. cereviciae}} & \multicolumn{2}{c|}{\textit{B. subtillus}}\tabularnewline
				\hline 
				$n$  & $C_{n}^{+}$  & $C_{n}^{-}$  & $C_{n}^{+}$  & $C_{n}^{-}$\tabularnewline
				\hline 
				\hline 
				3  & 35  & 27  & 129  & 75\tabularnewline
				\hline 
				4  & 1174  & 114  & 1366  & 547\tabularnewline
				\hline 
				5  & 152  & 135  & 1780  & 1031\tabularnewline
				\hline 
				6  & 3855  & 763  & 8003  & 4975\tabularnewline
				\hline 
				7  & 1875  & 1119  & 20645  & 14261\tabularnewline
				\hline 
				8  & 34321  & 13332  & 72722  & 48309\tabularnewline
				\hline 
				9  & 29473  & 16740  & 179137  & 122709\tabularnewline
				\hline 
				10  & 271954  & 128469  & 547246  & 364806\tabularnewline
				\hline 
				11  & 476800  & 279258  & 1443443  & 1002857\tabularnewline
				\hline 
			\end{tabular}
			\par\end{centering}
		\caption{Values of the number of positive $C_{n}^{+}$ and negative $C_{n}^{-}$ cycles of length $3\protect\leq n\protect\leq11$ for the gene regulatory networks of yeast and \textit{B. subtilis}.}
		
		\label{Cycles Altafini} 
	\end{table}
	
\end{document}